\documentclass[12pt]{article}
\usepackage[T1]{fontenc}
\usepackage{lmodern}
\usepackage[utf8]{inputenc}
\usepackage{microtype}
\usepackage{amsmath}
\usepackage{amssymb}
\usepackage{amsthm}
\usepackage{mathtools}
\usepackage{bm}
\usepackage{graphicx}
\usepackage{booktabs,array,tabularx}
\usepackage{caption}
\usepackage{placeins}
\theoremstyle{plain} 
\newtheorem{theorem}{Theorem}
\newtheorem{assumption}{Assumption}
\newtheorem{lemma}{Lemma}
\newtheorem{corollary}{Corollary}
\newtheorem{proposition}{Proposition}

\theoremstyle{remark} 
\newtheorem{remark}{Remark}
\theoremstyle{plain}
\newtheorem{atheorem}{Theorem}[section]
\newtheorem{aproposition}[atheorem]{Proposition}
\newtheorem{aassumption}[atheorem]{Assumption}

\usepackage[authoryear]{natbib}
\newcolumntype{Y}{>{\centering\arraybackslash}X}

\usepackage{geometry}
\usepackage{setspace}
\usepackage[colorlinks,
            linkcolor=blue,
            anchorcolor=blue,
            citecolor=blue]{hyperref}
\allowdisplaybreaks

\newcommand{\tableformat}{\centering\footnotesize
  \setlength{\tabcolsep}{3pt}\renewcommand{\arraystretch}{1.08}}
\newenvironment{papertabular}[1]
  {\begin{tabular*}{\linewidth}{@{\extracolsep{\fill}}#1@{}}}
  {\end{tabular*}}
\newenvironment{floatnotes}
  {\par\vspace{4pt}\noindent\begin{minipage}{\linewidth}\footnotesize\raggedright}
  {\end{minipage}\par}

\newcommand{\E}{\mathbb{E}}
\newcommand{\ISDp}{ISD$^{\ast}$}
\newcommand{\V}{\mathbb{V}}
\newcommand{\R}{\mathbb{R}}
\newcommand{\F}{\mathcal{F}}
\newcommand{\I}{\mathcal{I}}
\newcommand{\N}{\mathcal{N}}
\newcommand{\norm}[1]{\lVert #1\rVert}
\newcommand{\tr}{\operatorname{tr}}
\newcommand{\op}{\operatorname{op}}

\hypersetup{pdftitle={Tweedie's Formula and Score-Driven Updating},
 pdfauthor={Peter Reinhard Hansen and Chen Tong},
 pdfkeywords={Tweedie's formula, score-driven models, Bayesian filtering, exponential families, posterior mean, implicit updating}}

\makeatother

\title{Tweedie's Formula and Score-Driven Updating}

\author{Peter Reinhard Hansen$^{\mathsection}$\thanks{Corresponding author. Email addresses: hansen@unc.edu (P.~R.~Hansen); tongchen@xmu.edu.cn (C.~Tong).}\quad and\quad Chen Tong$^{\ddagger}$
\\[0.1cm] \small $^{\mathsection}$Department of Economics, University of North Carolina at Chapel Hill,
\\[-2mm] \small Chapel Hill, NC 27599, USA
\\[-2mm] \small $^{\ddagger}$School of Economics, Xiamen University, Xiamen, Fujian 361005, China}

\date{\small \today}

\begin{document}

\maketitle

\begin{abstract}
Tweedie's formula expresses a conditional mean in terms of a score function, and
this connects score-driven models with Bayesian updating. We show that the score-driven
update is identical to a Bayesian update in some special cases. More generally, we
show that score-driven updating is a local approximation to Bayesian updating, with
an error of second order in the predictive variance. The second-order approximation
error has two components, one from the curvature of the log likelihood and one
from the asymmetry of the posterior distribution. Implicit score-driven (ISD)
updating, which evaluates the score at the updated parameter, corrects for the
former but not the latter. A simple correction for the second component leads
to an improved third-order approximation, which we denote \ISDp{}. Numerical
experiments show substantial gains from ISD and \ISDp{}.
\end{abstract}
{\noindent}{\small\textit{Keywords:}}{\small{}
Tweedie's formula; score-driven models; Bayesian filtering; exponential families; posterior mean; implicit updating.}{\small\par}

{\noindent}{\small\textit{2020 Mathematics Subject Classification:}}{\small{}
Primary 62M20; Secondary 62M10, 62C10.}{\small\par}

{\noindent}{\small\textit{JEL Classification:}}{\small{}
C11, C13, C16, C22}{\small\par}

\setstretch{1.3}
\vspace{0.5\baselineskip}
\section{Introduction}\label{sec:intro}

Score-driven models update the time-varying parameters of a conditional distribution with the score of the conditional log likelihood; see \citet{CrealKoopmanLucas:2013} and \citet{Harvey:2013}, and \citet{ArtemovaBlasquesVanBrummelenKoopman:2022} for a recent overview. The models are observation driven in the terminology of \citet{Cox:1981}, so the likelihood is available by the prediction-error decomposition, and the score adapts the update to the local shape of the conditional density. In Student-$t$ location models, the score attenuates the influence of extreme observations; see \citet{HarveyLuati:2014}. Related robust location filters use bounded innovations obtained from a cumulative distribution function; see \citet{CataniaHarveyLuati:2026}. Earlier related work includes the approximate non-Gaussian filter of \citet{Masreliez:1975} and the observation-driven filters of \citet{HarveyFernandes:1989} for count and qualitative data, which use conjugate updating with discounting. Important theoretical justifications for score-driven updating are information-theoretic. \citet{BlasquesKoopmanLucas:2015} establish a local optimality property based on Kullback--Leibler criteria, \citet{GorgiLauriaLuati:2024} develop this approach further, and \citet{dePunderDimitriadisLange:2026} show that the appropriate guarantee is an expected Kullback--Leibler improvement. These results justify the direction of the update by showing that it improves the conditional model.

The starting point for this paper is Tweedie's formula, which relates a conditional mean to the score of a marginal density. Consider the convolution $Y=\mu+\varepsilon$, where $\varepsilon\sim\N(0,\sigma^2)$ is independent of $\mu$, and $\mu$ can have any distribution. Tweedie's formula states that
$$
\E(\varepsilon|Y=y)=-\sigma^2\partial_y\log f(y),
$$
where $f$ is the marginal density of $Y$, and equivalently $\E(\mu|Y=y)=y+\sigma^2\partial_y\log f(y)$; see \citet{Robbins:1956} and \citet{Efron:2011}. An analogous identity holds in parameter space, where the posterior mean of a time-varying parameter equals its predictive mean plus the predictive covariance times a score. So, Tweedie's formula shows that score-driven updating is closely related to Bayesian posterior-mean updating, although score-driven models are based on a different score.\footnote{To the best of our knowledge, Tweedie's formula has not previously been used to interpret score-driven updating.} The Bayesian update uses the score of the predictive density, whereas score-driven models use the score of the conditional density.

In this paper, we investigate this connection from the score-driven side. The key feature of Tweedie's formula is that only one of the two components in the convolution must be Gaussian, while the other can have any distribution. In the parameter-space version, the Gaussian component is the predictive distribution of the parameter, and the observation density can be arbitrary. Thus, a Bayesian posterior-mean update is a score update for any observation density, and the question is how it relates to the conditional-score update that score-driven models use. We establish the connection in three steps. For natural exponential families with conjugate priors, the two updates are identical in the expectation parameter, so the score-driven update is exactly Bayesian. Under a Gaussian predictive law, the predictive-marginal score is the posterior average of the conditional score, and the conditional score at the predictive mean is its leading term, so the score-driven update is a local approximation with an error of second order in the predictive covariance. The second-order error, in turn, explains the difference between explicit and implicit score-driven updating and leads to an improved implicit update. The exact case rests on conjugacy. In Bayesian analysis, a prior family is conjugate to a likelihood if the posterior is in the same family. An important class of likelihoods with conjugate priors is the natural exponential family, which includes the Gaussian location, Poisson, binomial, and gamma models. Its conjugate prior acts as $n$ pseudo-observations, and the posterior mean of the expectation parameter is a weighted average of its prior mean and the observation, with weight $1/(n+1)$ on the observation; see \citet{DiaconisYlvisaker:1979}. 

Our main contributions are as follows: First, we show that Tweedie's formula connects score-driven updating with Bayesian updating. For natural exponential families, the score-driven update of the expectation parameter is exactly the Bayesian posterior-mean update under a conjugate prior. The Kalman update, the variance update embedded in GARCH, and the discounted Poisson--gamma filter are special cases. Beyond the conjugate case, a Gaussian predictive law gives an exact correction equal to the predictive covariance multiplied by the score of the predictive marginal density with respect to the predictive mean of the parameter. Second, the conjugate case gives a Bayesian motivation for the scaling of the score. Score-driven models scale the score by $\I^{-d}$, where $\I$ is the Fisher information and $d\in\{0,\tfrac{1}{2},1\}$; see \citet{CrealKoopmanLucas:2013}. The conditional score divides the forecast error by its variance, and $d=1$ removes this normalization, so that the update is in terms of the forecast error, as in the conjugate posterior mean. The gain is $1/(n+1)$, where $n$ is the prior strength, and precision discounting makes it constant. A local version of this condition gives the same scaling beyond conjugacy. Third, we show that the covariance-scaled conditional score gives the leading term of this correction, with an error of second order in the predictive covariance. Fourth, we derive the second-order error under a Gaussian predictive law and use it to compare explicit and implicit score-driven updating with a common predictive covariance, which is established in Theorem \ref{thm:comparison}. Fifth, we propose a mean-corrected implicit score-driven update, \ISDp{}, which adds the leading correction from the posterior mode to the posterior mean and improves the error from second to third order. \ISDp{} has the lowest posterior-mean RMSE among ESD, ISD, \ISDp{}, and the explicit second-order update in all our numerical designs.

Explicit score-driven updating (ESD) evaluates the score at the predictive mean. Implicit score-driven updating (ISD) evaluates it at the updated parameter, so the correction and its evaluation point are determined jointly \citep{LangeVanOsVanDijk:2026,DonkervanHeelLangevanOsvanDijk:2026}. The second-order error of ESD has a curvature term, which involves the Hessian of the log likelihood, and a mean--mode term, which involves the third derivative and reflects posterior asymmetry. ISD is the posterior mode under a Gaussian prior. The mode captures the curvature term and is exact for quadratic log likelihoods. However, it does not account for the mean--mode term, so ESD and ISD have the same error order in general. \ISDp{} adds the leading Laplace correction from mode to mean, which removes the leading mean--mode term. ISD need not be closer to the posterior mean after every observation. In a Bernoulli log-odds example, ESD is closer after the more probable outcome. Averaged over the prior-predictive distribution, however, ISD has the lower risk in canonical exponential families when the predictive covariance is small.

The results are useful for three specification choices in score-driven models. The first is the scaling exponent $d$. The choice $d=1$ gives the exact update under conjugacy, but it amplifies observations with little information when the information in the observations varies over time, and Section~\ref{sec:index} discusses $d=0$ and $d=\tfrac{1}{2}$. The second is the parameterization. Updates of the expectation parameter have the exact Bayesian interpretation, whereas updates of a transformed parameter, such as the log variance, lose this interpretation but have scale-free innovations and may therefore be more stable in empirical studies. The third is the point at which the score is evaluated, which is the difference between ESD and ISD. The precision-discounting condition also gives the gain a Bayesian interpretation. The gain of the exact update is $1-\delta$, where $\delta$ is a discount factor. In the exact conjugate variance update, a gain of $0.05$ corresponds to $\delta=0.95$ and an effective memory of $(1-\delta)^{-1}=20$ observations, and Section~\ref{sec:example-variance} relates this gain to the coefficients of the GARCH(1,1) model.

Tweedie's formula was reported by \citet{Robbins:1956}, who credited personal correspondence with M.~C.~K.~Tweedie, and \citet{Efron:2011} emphasized its use in empirical Bayes and for selection bias. The conjugate posterior means in Section~\ref{sec:ExactResults} are classical; see \citet{WestHarrisonMigon:1985}, \citet{WestHarrison:1989}, and \citet{HarveyFernandes:1989}. What is new is that they are identical to inverse-Fisher-scaled score-driven updates, and Tweedie's formula extends this Bayesian interpretation beyond conjugate models. We build on classical identities and approximations from non-Gaussian filtering. \citet{West:1981} expressed posterior moments through derivatives of the predictive density and studied approximations based on likelihood derivatives, \citet{BuccheriBormettiCorsiLillo:2023} develop score- and Hessian-based perturbation filters, and \citet{LangeVanOsVanDijk:2026} formulate ISD through a penalized likelihood and show that it is a Gaussian-prior posterior mode. The correction from mode to mean is the classical Laplace approximation; see, e.g., \citet{TierneyKadane:1986}. Relative to this literature, our contributions are the Bayesian interpretation of score-driven updating and the comparison of approximation errors and posterior risks with a common posterior target and covariance multiplier.

The remainder of the paper is organized as follows. Section~\ref{sec:tweedie} presents Tweedie's formula for Gaussian and natural exponential families. Section~\ref{sec:index} separates the normalization in the score from the scaling of the score. Section~\ref{sec:ExactResults} treats the conjugate case, in which score-driven updating is exact. Section~\ref{sec:local} develops the local approximation theory, compares ESD and ISD, derives the mean correction, and connects the comparison to covariance scaling and recursive propagation. Section~\ref{sec:examples} presents the examples, Section~\ref{sec:numerics} reports numerical evidence, and Section~\ref{sec:conclusion} concludes. Proofs, additional results, and implementation details are in the appendix.

\section{Tweedie's Formula}\label{sec:tweedie}

We present Tweedie's formula for the Gaussian location model, where it takes its simplest form, and for natural exponential families, where related identities give the posterior means of the natural and expectation parameters. The latter identity leads to the exact filtering result in Section~\ref{sec:ExactResults}.

\subsection{The Gaussian location case}

Consider a convolution of a random variable, $\mu$, with density $\pi$ and an independent Gaussian $\varepsilon\sim\N(0,\sigma^2)$. Tweedie's formula expresses the conditional mean of $\mu$ given $Y=\mu+\varepsilon$ in terms of the score of the marginal density of $Y$. 
Let
$
f(y)=\int \phi_\sigma(y-\mu)\pi(\mu)d\mu
$
denote the marginal density of $Y$, where $\phi_\sigma$ is the $\N(0,\sigma^2)$ density. Because $\phi_\sigma$ is positive and has a bounded derivative, $f$ is positive and differentiable for every prior, and no further conditions are needed.

\begin{lemma}[Tweedie's formula]
\label{lem:tweedie_gaussian}
For every $y$,
$$
\E[\mu|Y=y]
=
y+\sigma^2\partial_y\log f(y).
$$
Equivalently, $\E[\varepsilon|Y=y]=-\sigma^2\partial_y\log f(y)$.
\end{lemma}

Lemma \ref{lem:tweedie_gaussian} holds for any distribution of $\mu$. The distribution of $\mu$ enters only through $f$.\footnote{For example, if $\mu$ is Cauchy, then $Y$ has a Voigt density, and Lemma \ref{lem:tweedie_gaussian} gives $\E[\varepsilon|Y=y]=-\sigma^2\partial_y\log f(y)$. In \citet{HansenTong:2026Voigt}, this identity yields the conditional Gaussian component in the Gauss--Cauchy convolution and produces a redescending update.} Hence, the posterior mean can be computed from the marginal density or its score, and an estimate of the marginal score gives an empirical-Bayes estimate of the posterior mean; see \citet{Efron:2011}. In dynamic models, the marginal score requires integration over the latent state and is typically unavailable, and score-driven models use instead the conditional likelihood score. We compare these two scores throughout the paper.

\subsection{Natural exponential families}

Lemma~\ref{lem:tweedie_gaussian} is a special case of an identity for natural exponential families, which \citet{Efron:2011} attributes to \citet{Robbins:1956}; see also \citet{Barndorff-Nielsen:1978book}. Suppose that the distribution of $Y$ given $\theta$ belongs to a natural exponential family (NEF),
$$
p(y|\theta)=h(y)\exp\{\theta y-\psi(\theta)\},
$$
where $\theta$ is the natural parameter and $\psi$ is the cumulant generating function. Some well-known NEFs are listed in Table~\ref{tab:nef-examples}. Let
$
f(y)=\int p(y|\theta)\pi(\theta)d\theta
$
denote the marginal density of $Y$ induced by a prior distribution $\pi$ for $\theta$.

\begin{table}[!htbp]
\tableformat
\caption{Examples of natural exponential families}
\label{tab:nef-examples}
\begin{papertabular}{llllll}
\toprule
Distribution & Known & $\theta$ & $\psi(\theta)$ & $\operatorname{var}(Y)$ & Conjugate prior \\
\midrule
Normal $\N(\mu,\sigma^2)$ & $\sigma^2$ & $\mu/\sigma^2$ & $\sigma^2\theta^2/2$ & $\sigma^2$ & Normal on $\mu$ \\
Poisson $(\mu)$ & None & $\log\mu$ & $e^\theta$ & $\mu$ & Gamma on $\mu$ \\
Binomial $(N,p)$ & $N$ & $\log\{p/(1-p)\}$ & $N\log(1+e^\theta)$ & $\mu(1-\mu/N)$ & Beta on $p$ \\
Negative binomial $(r,p)$ & $r$ & $\log p$ & $-r\log(1-e^\theta)$ & $\mu(1+\mu/r)$ & Beta on $p$ \\
Gamma $(k,\mu)$ & $k$ & $-k/\mu$ & $-k\log(-\theta)$ & $\mu^2/k$ & Inverse gamma on $\mu$ \\
Inverse Gaussian $(\mu,\lambda)$ & $\lambda$ & $-\lambda/(2\mu^2)$ & $-(-2\lambda\theta)^{1/2}$ & $\mu^3/\lambda$ & Non-standard \\
\bottomrule
\end{papertabular}
\begin{floatnotes}
Notes: Each distribution is written as $p(y|\theta)=h(y)\exp\{\theta y-\psi(\theta)\}$ with the listed parameter known, so that the mean is $\mu=\psi'(\theta)$ and the variance is $\psi''(\theta)$. The gamma distribution has shape $k$ and mean $\mu$. Note that the Bernoulli, geometric, and exponential distributions are nested as special cases.
\end{floatnotes}
\end{table}

\begin{lemma}[Tweedie's formula for NEFs]
\label{lem:tweedie_nef}
Suppose that $h$ is positive and differentiable on an open neighborhood $U$ of $y$, and that there is an integrable function $H$ such that, for all $u\in U$,
$\left|\theta\exp\{\theta u-\psi(\theta)\}\pi(\theta)\right|\leq H(\theta)$ for almost all $\theta$. Then
$$
\E[\theta|Y=y]
=
\partial_y\log\frac{f(y)}{h(y)}
=
\partial_y\log f(y)-\partial_y\log h(y).
$$
\end{lemma}

The posterior mean of the natural parameter differs from the marginal score unless $\partial_y\log h(y)=0$. In the Gaussian location model, the natural parameter is $\theta=\mu/\sigma^2$ and
$h(y)\propto\exp\{-y^2/(2\sigma^2)\}$, so Lemma \ref{lem:tweedie_nef} gives
$\E[\mu/\sigma^2|Y=y]
=
\partial_y\log f(y)+y/\sigma^2$, and multiplying by $\sigma^2$ gives the result in Lemma \ref{lem:tweedie_gaussian}.

\begin{lemma}
\label{lem:nef_expectation_parameter_identity}
Consider the natural exponential family
$$
p(y|\theta)=h(y)\exp\{\theta y-\psi(\theta)\},
\qquad
\mu=\psi^\prime(\theta),
$$
and let $\pi(\theta)$ be a differentiable prior density on the natural parameter space. Suppose the posterior kernel
$\exp\{\theta y-\psi(\theta)\}\pi(\theta)$ vanishes at the endpoints of the natural parameter space, its derivative is integrable, and the expectations below are finite. Then
$$
\E(\mu|Y=y)
=
y+\E\{\partial_\theta\log \pi(\theta)|Y=y\}.
$$
In particular, if $\pi(\theta)\propto\exp\{\tau\theta-n\psi(\theta)\}$, where $n>0$, then
$
\E(\mu|Y=y)
=
(\tau+y)/(n+1).
$
\end{lemma}

The conjugate case of Lemma \ref{lem:nef_expectation_parameter_identity} is classical; see \citet{DiaconisYlvisaker:1979}, who also show that, when the support of the observation is sufficiently rich, a posterior mean of $\mu$ that is linear in $y$ characterizes the conjugate prior. Lemmas \ref{lem:tweedie_nef} and \ref{lem:nef_expectation_parameter_identity} concern different posterior means, and in general $\E\{\psi'(\theta)|y\}\neq\psi'\{\E(\theta|y)\}$. We use Lemma \ref{lem:nef_expectation_parameter_identity} in Section~\ref{sec:ExactResults}, because inverse-Fisher scaling turns the conditional score into the forecast error $y-\mu$ in the expectation parameterization. For discrete natural exponential families, the derivative with respect to $y$ is replaced by a probability ratio; for example, if $Y|\mu\sim\operatorname{Poisson}(\mu)$ and $f$ is the marginal probability mass function, then $\E(\mu|Y=y)=(y+1)f(y+1)/f(y)$ whenever $f(y)>0$. This affects only the marginal Tweedie identity, which differentiates with respect to the observation. It does not affect the conditional likelihood score, which differentiates with respect to a continuous parameter and is the score used in score-driven models.

\section{Variance Functions and the Scaling of the Score}\label{sec:index}

The conditional score normalizes the forecast error by its variance. This normalization is part of the likelihood score, and it differs from the scaling that a score-driven model applies to the score. This section separates the two, and we use scalar notation for simplicity. For an exponential dispersion model, let $\mu$ denote the conditional mean and $\phi$ the dispersion parameter, so that
$
\operatorname{var}(Y|\mu)=\phi V(\mu)
$
with variance function $V(\mu)$.

\begin{lemma}
\label{lem:variance_function_score}
For a regular exponential dispersion model with conditional mean $\mu$ and variance function $V(\mu)$,
$$
\partial_\mu\ell(y;\mu)=\frac{y-\mu}{\phi V(\mu)},
\qquad
\I_\mu(\mu)=\frac{1}{\phi V(\mu)}.
$$
Consequently, $\I_\mu(\mu)^{-1}\partial_\mu\ell(y;\mu)=y-\mu$.
\end{lemma}

Thus, the likelihood score divides the forecast error $y-\mu$ by the variance $\phi V(\mu)$, and inverse-Fisher scaling undoes this normalization.

The Tweedie class is the subclass of exponential dispersion models with variance function
$V(\mu)=\mu^p$, where the scalar $p$ is the Tweedie index \citep{Tweedie:1984,Jorgensen:1987}. The normal, Poisson, gamma, and inverse Gaussian distributions in Table~\ref{tab:nef-examples} correspond to $p=0,1,2$, and $3$, respectively. Lemma \ref{lem:variance_function_score} gives
$\partial_\mu\ell(y;\mu)=(y-\mu)/(\phi\mu^p)$ and $\I_\mu(\mu)=1/(\phi\mu^p)$.
If a score-driven recursion scales the mean score by $\I_\mu(\mu)^{-d}$, then
$$
\I_\mu(\mu)^{-d}\partial_\mu\ell(y;\mu)
=
\frac{y-\mu}{\phi^{1-d}\mu^{p(1-d)}}.
$$
Thus, the index $p$ determines the normalization in the score, and the exponent $d$ determines how much of the normalization remains after scaling. With $d=0$ the update uses the score itself, with $d=1/2$ it uses a standardized forecast error, and with $d=1$ it uses the forecast error $y-\mu$. The case $d=1$ matches the exact conjugate update in Proposition \ref{prop:exact_nef_conjugate}.

The normalization also depends on the parameterization. For $\eta=\log\mu$, the score is
$$
\partial_\eta\ell(y;\mu)=\frac{\mu(y-\mu)}{\phi V(\mu)},
$$
which is $(y-\mu)/(\phi\mu^{p-1})$ for the Tweedie class. A Bayesian interpretation of the update must therefore specify its target, which can be the posterior mean of $\mu$, of $\log\mu$, or of another transformation.

Inverse-Fisher scaling can be a poor choice when the information in the observations varies over time, because it amplifies observations with little information. For example, in a binomial model with logit parameter $\eta_t$, the Fisher information is $\I_t=N_tp_t(1-p_t)$, so the conditional variance of $\I_t^{-d}s_t$ is proportional to $\I_t^{1-2d}$. Thus, $d=0$ keeps the information weighting of the likelihood, $d=1/2$ standardizes the score, and $d=1$ amplifies observations with low information. With declining risk sets, for example, a large scaling exponent makes the recursion sensitive to noisy observations.

\section{Exact Expectation-Space Filtering under Conjugacy}\label{sec:ExactResults}

This section shows that, for NEF observation densities, the inverse-Fisher-scaled score-driven update of the expectation parameter has an exact Bayesian counterpart. It is the posterior mean under a conjugate prior, and the gain determines the prior strength. We combine the expectation-parameter identity in Lemma \ref{lem:nef_expectation_parameter_identity} with the variance-function normalization in Lemma \ref{lem:variance_function_score} in a filtering problem with a natural exponential-family observation density. Conjugacy is expressed in the natural parameter $\theta_t$, but the score representation of the Bayesian update is simplest in the expectation parameter $\mu_t=\psi^\prime(\theta_t)$. The boundary condition of Lemma \ref{lem:nef_expectation_parameter_identity} holds in the usual conjugate families.

\begin{proposition}[Conjugate updating as a score update]
\label{prop:exact_nef_conjugate}
Let
$$
p_t(y_t|\theta_t)
=
h(y_t)\exp\{\theta_t y_t-\psi(\theta_t)\},
\qquad
\mu_t=\psi^\prime(\theta_t).
$$
Suppose the predictive prior density of $\theta_t|\F_{t-1}$ is proper and
conjugate,
$$
\pi(\theta_t|\F_{t-1})
\propto
\exp\{\tau_{t|t-1}\theta_t-n_{t|t-1}\psi(\theta_t)\},
\qquad
n_{t|t-1}>0,
$$
and that the boundary and integrability conditions in Lemma \ref{lem:nef_expectation_parameter_identity} hold for both the prior and the posterior. Then the filtered posterior is conjugate, with
$
\tau_{t|t}=\tau_{t|t-1}+y_t$ and $n_{t|t}=n_{t|t-1}+1,
$
and the exact Bayesian posterior mean in expectation space satisfies
\begin{equation}\label{eq:conjugate-score}
\mu_{t|t}
=
\mu_{t|t-1}
+
\frac{1}{n_{t|t-1}+1}
\I_\mu(\mu_{t|t-1})^{-1}
\left.
\partial_\mu\log p_t(y_t|\mu)
\right|_{\mu=\mu_{t|t-1}}.
\end{equation}
If the local precision-discounting condition $n_{t|t-1}=\delta n_{t|t}$ with $0<\delta<1$ holds, then $n_{t|t-1}=\delta/(1-\delta)$, and the gain in (\ref{eq:conjugate-score}) is $1-\delta$.
\end{proposition}

The representation (\ref{eq:conjugate-score}) combines a natural exponential-family observation density, a conjugate prior, updating of the expectation parameter, and inverse-Fisher scaling. The local precision-discounting condition makes the gain constant, $1-\delta$. Without it, (\ref{eq:conjugate-score}) remains exact with gain $\alpha_t=1/(n_{t|t-1}+1)$. Suppose that the next predictive prior remains proper and conjugate, satisfies the conditions of Proposition~\ref{prop:exact_nef_conjugate}, and has strength $n_{t+1|t}=\delta_t(n_{t|t-1}+1)$, where $0<\delta_t<1$ is measurable with respect to $\F_t$. The predictable gain then obeys $\alpha_{t+1}=\alpha_t/(\delta_t+\alpha_t)$, with fixed point $1-\delta$ when $\delta_t=\delta$. Conversely, a predictable gain sequence with $0<\alpha_t<1$ admits this discounting interpretation only if $\delta_t=\alpha_t(1-\alpha_{t+1})/\alpha_{t+1}$ lies in $(0,1)$. For an admissible sequence, a larger next-period gain corresponds to stronger discounting at a given current gain. This provides an interpretation of adaptive gains of the kind motivated by \citet{BlasquesGorgiKoopman:2019}, who increase the gain when recent scores are positively autocorrelated; it does not establish an exact Bayesian representation for every accelerating score-driven specification. Updating a natural or transformed parameter, or using a different score scaling, generally loses the exact equivalence, although the update may still be a useful local approximation.

Conjugate updating with discounting is classical in dynamic generalized linear models, where the discount factor $\delta$ determines how fast past observations are down-weighted; see, e.g., \citet{WestHarrisonMigon:1985} and \citet[Ch.~14]{WestHarrison:1989}. The discounting recursion $n_{t+1|t}=\delta(n_{t|t-1}+1)$ has the fixed point $\delta/(1-\delta)$, and our condition imposes this value at every date. \citet{HarveyFernandes:1989} developed the observation-driven version for count and qualitative data, in which discounted conjugate updating gives exponentially weighted forecasts. Their Poisson--gamma filter is the Poisson example in Section~\ref{sec:examples}. Proposition \ref{prop:exact_nef_conjugate} does not give a new recursion. It shows that the conjugate update is an inverse-Fisher-scaled conditional-score update in expectation space. The proposition concerns the posterior mean before any transition dynamics are applied. A recursion that adds an intercept or autoregressive terms to the correction gives a rule for the next predictive mean, not a posterior mean.

\citet{DonkervanHeelShephard:2026} derive a related filter by maximizing an exponentially weighted log-likelihood with anchoring. To compare their criterion with conjugate updating, restrict attention to a scalar sufficient statistic $y_j$ and a common, time-invariant cumulant $\psi$. Put $z_j=wy_j+(1-w)\E(Y_j)$, where $w\in[0,1]$ and $1-w$ is the anchoring weight. Their criterion then has the form $\theta m_t-n_t\psi(\theta)$, with $m_t=\sum_{j=1}^t\delta^{t-j}z_j$ and $n_t=\sum_{j=1}^t\delta^{t-j}$. It defines a conjugate kernel $q_t(\theta)\propto\exp\{m_t\theta-n_t\psi(\theta)\}$. If this kernel is proper, satisfies the boundary and integrability conditions of Lemma~\ref{lem:nef_expectation_parameter_identity}, and has an interior mode $\tilde\theta_t$, then $\E_{q_t}\{\psi'(\theta)\}=\psi'(\tilde\theta_t)=m_t/n_t$. Thus, the value of the expectation parameter at the natural-parameter mode equals its mean under $q_t$. Writing this ratio as $\tilde\mu_t$, its recursion is $\tilde\mu_t=\tilde\mu_{t-1}+(z_t-\tilde\mu_{t-1})/n_t$ for $t\geq2$, with $\tilde\mu_1=z_1$. It uses the anchored pseudo-observation $z_t$, rather than $y_t$ in Proposition~\ref{prop:exact_nef_conjugate}. The zero-strength initialization underlying these sums is an algebraic convention, outside that proposition's proper-prior assumptions; a Bayesian interpretation requires a proper initial prior or a separately justified limiting argument. Moreover, the data enter the kernel through likelihood powers $w\delta^{t-j}$, so this is a power-likelihood interpretation rather than automatically the ordinary posterior for the original observations. The anchor contributes a conjugate kernel of strength $(1-w)n_t$, centered at $\sum_{j=1}^t\delta^{t-j}\E(Y_j)/n_t$; this center is the unconditional mean when $\E(Y_j)$ is constant. In this common-cumulant case, recovering $\tilde\theta_t$ requires inverting $\psi'$, whereas the expectation parameter is available directly. The authors also show that inverse-Fisher-scaled updates of the natural parameter reproduce their predictor in the Gaussian case but not in general. In expectation space, Lemma~\ref{lem:variance_function_score} makes the scaled score equal to $y_t-\mu_{t|t-1}$. A score-driven recursion with an intercept and matching coefficients therefore reproduces their steady-state predictor in the common-cumulant, stationary setting.

\section{Local Score-Driven Approximations to Bayesian Filtering}\label{sec:local}

The exact results established above rely on restrictive assumptions. In applied work, most models do not have the conjugate structure, and their posterior mean is not available in closed form. Examples are the Student-$t$ location model and models for a log variance, log intensity, or log odds with a Gaussian predictive law, which we consider below. Under a Gaussian predictive law, the exact posterior correction is the predictive covariance multiplied by the score of the predictive marginal density with respect to the predictive mean. Replacing this marginal score by the conditional score at the predictive mean gives the leading approximation to the correction. Its second-order error allows us to compare explicit and implicit updating under a common Gaussian predictive law.

\subsection{The predictive-marginal correction}

Consider a general dynamic model in which
$$
y_t|\theta_t\sim p_t(y_t|\theta_t),
$$
where $\theta_t\in\R^d$ is a time-varying parameter and
$\F_{t-1}$ denotes the information available before observing
$y_t$. In a parameter-driven formulation, $\theta_t$ is stochastic
conditional on $\F_{t-1}$. We write
$$
a_t
=
\E(\theta_t|\F_{t-1}),
\qquad
P_t
=
\operatorname{var}(\theta_t|\F_{t-1})
$$
for its predictive mean and covariance. Bayes' rule gives
$\pi(\theta_t|y_t,\F_{t-1})
\propto
p_t(y_t|\theta_t)\pi(\theta_t|\F_{t-1})$, and the filtered estimate is
$m_t=\E(\theta_t|y_t,\F_{t-1})$.

Proposition \ref{prop:parameter_space_tweedie} is a parameter-space analogue of Tweedie's formula. It allows for a general conditional density $p_t(y_t|\theta_t)$. The only condition is on the Bayesian side, where the predictive distribution of $\theta_t$ must belong to a natural exponential family. From Section~\ref{sec:local-approx} onward, this is the Gaussian law $\N(a_t,P_t)$.

\begin{proposition}
\label{prop:parameter_space_tweedie}
Let $y_t|\theta_t\sim p_t(y_t|\theta_t)$, where $\theta_t\in\R^d$.
Suppose that, conditionally on $\F_{t-1}$, the predictive density of
$\theta_t$ belongs to a regular natural exponential family with identity
natural statistic:
$$
\pi(\theta|\F_{t-1};\eta_t)
=
h_t(\theta)\exp\{\eta_t^\prime\theta-A_t(\eta_t)\},
$$
where $a_t=\E(\theta_t|\F_{t-1};\eta_t)=\nabla_{\eta}A_t(\eta_t)$, and $P_t=\operatorname{var}(\theta_t|\F_{t-1};\eta_t)=\nabla_{\eta}^2A_t(\eta_t)$ is positive definite. Define the predictive marginal density by
$$
f_t(y_t;\eta_t)
=
\int p_t(y_t|\theta)\pi(\theta|\F_{t-1};\eta_t)d\theta .
$$
Suppose that $f_t(y_t;\eta_t)>0$, that the posterior first moment is finite, and that differentiation with respect to
$\eta_t$ may be passed under the integral sign. Then
$$
\nabla_{\eta_t}\log f_t(y_t;\eta_t)
=
\E(\theta_t|y_t,\F_{t-1})-a_t.
$$
Equivalently, under the locally invertible mean parameterization,
$$
\E(\theta_t|y_t,\F_{t-1})-a_t
=
P_t\nabla_{a_t}\log f_t(y_t;a_t).
$$
\end{proposition}

In the second expression, the derivative is taken within the exponential family, so $P_t$ generally changes with $a_t$. For a Gaussian predictive law, $P_t$ can be held fixed while $a_t$ varies. Under $\theta_t|\F_{t-1}\sim\N_d(a_t,P_t)$, write
$s_t(\theta)=\nabla_\theta\log p_t(y_t|\theta)$ for the conditional score. If the likelihood is differentiable, the relevant expectations are finite, and the Gaussian integration-by-parts boundary terms vanish, then
\begin{equation}\label{eq:gaussian-tweedie}
 m_t-a_t=P_t\nabla_{a_t}\log f_t(y_t;a_t,P_t)
        =P_t\E\{s_t(\theta_t)| y_t,\F_{t-1}\}.
\end{equation}
The first equality is Tweedie's formula for a Gaussian predictive law, with the derivative taken with respect to $a_t$ for fixed $P_t$. Gaussian integration by parts gives the second equality, which shows that the marginal score is the posterior average of the conditional score. Score-driven updating replaces this posterior average by the conditional score at a single parameter value. In the linear Gaussian state-space model, the first equality in (\ref{eq:gaussian-tweedie}) is the Kalman update written in score form.

The predictive marginal score requires an integral over the predictive distribution and is rarely available in closed form, so standard score-driven models use the conditional score $s_t(a_t)$ in its place. \citet{KoopmanLucasScharth:2016} find that observation-driven score models can forecast almost as well as correctly specified parameter-driven models, even when the latter generate the data. Proposition~\ref{prop:small_variance_score_expansion} below provides a local explanation for this finding: the covariance-scaled conditional score gives the leading term of the exact correction under a concentrated Gaussian predictive law.

Differentiating once more gives the posterior variance. If second derivatives may be taken under the integral and posterior second moments are finite, then
\begin{equation}\label{eq:gaussian-cov-identity}
 V_t:=\operatorname{var}(\theta_t| y_t,\F_{t-1})
 =P_t+P_t\nabla_{a_t}^2\log f_t(y_t;a_t,P_t)P_t,
\end{equation}
so the Hessian of the log predictive density determines how the observation changes the posterior variance. We use (\ref{eq:gaussian-cov-identity}) in Section~\ref{sec:propagation}, where it determines how the posterior mean responds to the predictive mean.

\subsection{The local score approximation}\label{sec:local-approx}

We now show that the predictive covariance multiplied by the conditional score at the predictive mean is the leading term of the exact correction in (\ref{eq:gaussian-tweedie}), with an error of second order in the predictive covariance. All local expansions in this section use the following assumption, which combines smoothness with a global envelope that controls the Gaussian tails as the predictive covariance decreases. The order $K$ of the envelope is specified in each result.

\begin{assumption}[Local Gaussian expansion of order $K$]\label{ass:local}
For the fixed observation $y_t$ in the likelihood support, $L(\theta)=p_t(y_t|\theta)$ is positive and $\max(K,3)$ times continuously differentiable on $\R^d$. There are finite constants $C$ and $c$, which may depend on $y_t$ and $a_t$ but not on $P_t$, such that, for $k=0,\ldots,K$ and all $x\in\R^d$,
\begin{equation}\label{eq:local-envelope}
 \norm{\nabla^k L(a_t+x)}\leq C L(a_t)\exp(c\norm{x}^2).
\end{equation}
When $P_t=\varepsilon\Omega$, the matrix $\Omega\succ0$ is fixed. For uniform assertions over sets of $(a_t,y_t,\Omega)$, the same $C$ and $c$ apply and the eigenvalues of $\Omega$ are bounded away from zero and infinity.
\end{assumption}

The case $K=0$ is an envelope condition on the likelihood itself. It permits quadratic growth of the log likelihood in the parameter, as in the Gaussian location model, and it holds whenever the likelihood is bounded in $\theta$. The derivative envelopes for $k\geq1$ make the likelihood derivatives integrable under the Gaussian predictive density once $2c\lambda_{\max}(P_t)<1$, and justify integration by parts. Assumption~\ref{ass:local} holds with $K=5$ for every model that we combine with a Gaussian predictive law in Sections~\ref{sec:examples} and \ref{sec:numerics}.

\begin{proposition}
\label{prop:small_variance_score_expansion}
Suppose that
$$
\theta_t|\F_{t-1}\sim \N(a_t,P_t),
$$
where $P_t$ is positive definite. Let
$\ell_t(\theta)=\log p_t(y_t|\theta)$ and
$s_t(a_t)=\nabla_\theta\ell_t(\theta)|_{\theta=a_t}$, and suppose Assumption \ref{ass:local} holds with $K=0$. Then, as $\|P_t\|\to0$,
$$
P_t\nabla_{a_t}\log f_t(y_t;a_t,P_t)
=
P_t s_t(a_t)+O(\|P_t\|^2),
$$
where the gradient is taken with respect to $a_t$, holding $P_t$ fixed. Consequently,
$$
\E(\theta_t|y_t,\F_{t-1})
=
a_t+P_t s_t(a_t)+O(\|P_t\|^2).
$$
\end{proposition}

If the state is the expectation parameter of an exponential dispersion model and $P_t=\kappa_P\I_\mu(a_t)^{-1}$, Lemma \ref{lem:variance_function_score} shows that the leading correction is $\kappa_P(y_t-a_t)$, which is the local counterpart of the conjugate update in Proposition \ref{prop:exact_nef_conjugate}. For constrained expectation parameters, the score-scaling identity remains valid, but the Gaussian approximation requires an unconstrained parameterization, as in the log-scale examples below. Section~\ref{sec:fisher} gives a condition under which $P_t$ is proportional to the inverse Fisher information.

\subsection{Explicit and implicit updating}\label{sec:esd-isd}

The correction $P_ts_t(a_t)$ evaluates the conditional score at the predictive mean. \citet{LangeVanOsVanDijk:2026} propose evaluating the score at the updated parameter instead, and \citet{DonkervanHeelLangevanOsvanDijk:2026} analyze both updates under misspecification. With a common predictive mean and covariance, we define explicit score-driven updating (ESD) and implicit score-driven updating (ISD) by
\begin{equation}\label{eq:two-updates}
 \Delta_t^E=P_t s_t(a_t),\qquad
 \Delta_t^I=P_t s_t(a_t+\Delta_t^I),
\end{equation}
and we write $\Delta_t^B=m_t-a_t$ for the exact correction under the Gaussian predictive law. ESD replaces the posterior average of the score in (\ref{eq:gaussian-tweedie}) by the score at $a_t$. ISD evaluates the score at the updated parameter, so the correction and its evaluation point are determined jointly. Both updates use the same $P_t$ and differ only in where the score is evaluated. In the Gaussian location model, ISD gives the Kalman update exactly, whereas ESD gives its first-order approximation; see Section~\ref{sec:examples}. With a different gain, an explicit update can also be exact, as in Section~\ref{sec:ExactResults}. The comparison here holds the covariance fixed.

\begin{proposition}
\label{prop:implicit}
Let $\ell_t(\theta)=\log p_t(y_t|\theta)$ be twice continuously
differentiable. If
$q_t(\theta)=\ell_t(\theta)-\tfrac{1}{2}(\theta-a_t)^\prime
P_t^{-1}(\theta-a_t)$ is strictly concave and tends to $-\infty$ as
$\norm{\theta}\to\infty$, it has a unique maximizer
$\widehat\theta_t^I=a_t+\Delta_t^I$, which satisfies
(\ref{eq:two-updates}). It is the Gaussian-prior posterior mode.
If the log likelihood is globally quadratic with Hessian $H_t$ and
$P_t^{-1}-H_t\succ0$, the posterior is Gaussian and
\begin{equation}\label{eq:quadratic-exact}
 \Delta_t^B=\Delta_t^I=(P_t^{-1}-H_t)^{-1}s_t(a_t).
\end{equation}
More generally, with $\overline H_t(d)=\int_0^1\nabla^2\ell_t(a_t+ud)du$, the implicit equation can be written exactly as
\begin{equation}\label{eq:resolvent}
 \Delta_t^I
 =\{P_t^{-1}-\overline H_t(\Delta_t^I)\}^{-1}s_t(a_t),
\end{equation}
whenever the inverse exists.
\end{proposition}

The mode interpretation and the average-Hessian representation (\ref{eq:resolvent}) are due to \citet{LangeVanOsVanDijk:2026}. Under the coercivity condition in Proposition~\ref{prop:implicit}, a sufficient condition for a unique mode is $P_t^{-1}-\nabla^2\ell_t(\theta)\succ0$ for every $\theta$. If the mode is not unique, our local results concern the mode whose displacement tends to zero with $P_t$. Equation (\ref{eq:resolvent}) shows that ISD adjusts the score by the average Hessian along the displacement, whereas ESD uses the score at $a_t$ for the whole displacement. Equation~(\ref{eq:quadratic-exact}) explains why ISD can be accurate even when $P_t$ is large. However, the posterior mode and mean differ when the posterior is asymmetric, so the curvature adjustment alone does not determine which update is closer to the posterior mean.

\subsection{Second-order errors and posterior risk}

Proposition~\ref{prop:small_variance_score_expansion} shows that ESD and the exact correction agree to first order. To compare ESD with ISD, we need the second-order term, because this is where the two updates differ. Theorem \ref{thm:comparison} derives this term and shows which part of it ISD captures. We suppress the time index and let $P_\varepsilon=\varepsilon\Omega$, where $\Omega\succ0$ is fixed and $\varepsilon\downarrow0$. Write $s=\nabla\ell(a)$, $H=\nabla^2\ell(a)$, and $T=\nabla^3\ell(a)$. For a matrix $M$, define $(T:M)_i=\sum_{j,k}T_{ijk}M_{jk}$, which simplifies to $TM$ in one dimension.

\begin{theorem}[Local Bayesian comparison of ESD and ISD]
\label{thm:comparison}
Suppose Assumption~\ref{ass:local} holds with $K=5$ and
$\theta\sim\N_d(a,\varepsilon\Omega)$ before observing $y$.
Define $u=\Omega H\Omega s$ and $v=\tfrac{1}{2}\Omega(T:\Omega)$.
For the local implicit solution that converges to zero,
\begin{align*}
 \Delta_\varepsilon^B
    &=\varepsilon\Omega s+\varepsilon^2(u+v)+O(\varepsilon^3),
       \\*
 \Delta_\varepsilon^E&=\varepsilon\Omega s,\\*
 \Delta_\varepsilon^I&=\varepsilon\Omega s+\varepsilon^2u
                         +O(\varepsilon^3).
\end{align*}
Consequently, the respective posterior-mean errors of ESD and ISD are
$\varepsilon^2(u+v)+O(\varepsilon^3)$ and
$\varepsilon^2v+O(\varepsilon^3)$. The statements are uniform under the
uniform conditions in Assumption~\ref{ass:local}.
\end{theorem}

ESD misses both second-order terms, $u+v$. ISD captures the curvature term $u$ and misses only the mean--mode term $v$. Hence both updates have errors of order $\varepsilon^2$ in general, and ESD can be closer to the posterior mean when $u$ and $v$ offset each other. The curvature term involves the Hessian and reflects the change in the score along the displacement. The mean--mode term involves the third derivative and is the leading term of the distance from the posterior mode to the posterior mean. With $\mathcal R_\varepsilon(D)=\E[\norm{\theta-a-D}^2| y]$, the posterior variance decomposition gives $\mathcal R_\varepsilon(D)-\mathcal R_\varepsilon(\Delta^B)=\norm{D-\Delta^B}^2$, so the ranking of two updates by posterior risk is their ranking by squared distance from the posterior mean.

\begin{corollary}\label{cor:conditional-risk}
Under Theorem~\ref{thm:comparison},
\begin{equation}\label{eq:conditional-risk}
 \mathcal R_\varepsilon(\Delta^E)-\mathcal R_\varepsilon(\Delta^I)
 =\varepsilon^4\{\norm{u}^2+2u^\prime v\}+O(\varepsilon^5).
\end{equation}
ISD has smaller risk for sufficiently small $\varepsilon$ when the displayed coefficient is positive, and ESD has smaller risk when it is negative. A zero coefficient requires higher-order analysis.
\end{corollary}

For example, the leading comparison favors ISD if $u\ne0$ and $u^\prime v\geq0$, or if $\norm{v}<\norm{u}/2$. In the quadratic case, the asymmetry term vanishes and ISD equals the posterior mean exactly. The Bernoulli example in Section~\ref{sec:example-bernoulli} shows that the sign of the coefficient can depend on the realized observation.

\subsection{Mean-corrected implicit updating}

From Theorem~\ref{thm:comparison} it follows that ISD misses part of the second-order term, namely the mean--mode term $v$. This term is known, so it is easy to correct for. Since ISD is the posterior mode, adding an estimate of the distance to the posterior mean moves the update toward the mean. Write $T^I=\nabla^3\ell(a+\Delta^I)$ and $W^I=\{P^{-1}-\nabla^2\ell(a+\Delta^I)\}^{-1}$, assuming that the latter matrix is positive definite, and define
\begin{equation}\label{eq:mcisd}
 \Delta^{I\ast}=\Delta^I+\frac12W^I(T^I:W^I).
\end{equation}
The added term is the leading Laplace correction from mode to mean; see, e.g., \citet{TierneyKadane:1986}. We call (\ref{eq:mcisd}) the mean-corrected implicit score-driven update, \ISDp{}.

\begin{proposition}[\ISDp{}]\label{prop:mcisd}
Under the conditions of Theorem~\ref{thm:comparison}, $W^I$ is positive definite for all sufficiently small $\varepsilon$, and
\begin{equation}\label{eq:mc-error}
 \Delta_\varepsilon^B-\Delta_\varepsilon^{I\ast}=O(\varepsilon^3).
\end{equation}
Under a globally quadratic log likelihood with $P^{-1}-H\succ0$, the adjustment vanishes and $\Delta^{I\ast}=\Delta^I=\Delta^B$.
\end{proposition}

Thus, the excess conditional risk of \ISDp{} is $O(\varepsilon^6)$, compared with $O(\varepsilon^4)$ for ESD and ISD in general, and for sufficiently small $\varepsilon$ \ISDp{} has strictly smaller risk than ESD when $u+v\ne0$ and than ISD when $v\ne0$. The explicit second-order update,
\begin{equation}\label{eq:second-order-filter}
 \Delta^{(2)}=Ps+PHPs+\frac12P(T:P),
\end{equation}
has the same local error order as \ISDp{}, but the two can differ at finite covariance, because \ISDp{} keeps the implicit curvature adjustment and evaluates the correction at the mode. The accuracy result for \ISDp{} is local, and the additive correction need not preserve the optimization or constraint-preserving properties of ISD. Section~\ref{sec:numerics} compares the updates at finite predictive covariances.

\subsection{Average risk and finite-covariance comparisons}
\label{sec:finite-risk}

The ranking of ESD and ISD can depend on the observation. Averaged over the prior-predictive distribution, however, ISD is more accurate in canonical exponential families when the predictive covariance is small. For the scalar case, let $p(y|\theta)=h(y)\exp\{\theta t(y)-b(\theta)\}$ have natural parameter space $\R$, let $\theta\sim\N(a,P)$, and let $\mathcal B_P^j=\E[(\theta-a-\Delta^j(Y))^2]$ denote the risk of rule $j$ under the joint prior-predictive law. All rules use the same $a$ and $P$. Averaging the local expansion over observations requires integrability conditions; Appendix C states them and proves the multivariate version of the result. The conditions hold for finite-support canonical families and for Poisson log intensity.

\begin{corollary}\label{cor:average-risk}
Under the integrability conditions in Appendix C, if $b''(a)>0$, then
\begin{equation*}
 \mathcal B_P^E-\mathcal B_P^I=P^4\{b''(a)\}^3+O(P^5)>0
\end{equation*}
for all sufficiently small $P>0$.
\end{corollary}

In a canonical family, the Hessian and the third derivative of the log likelihood do not depend on the observation. Under the prior-predictive distribution, the score evaluated at $a$ has mean $O(P)$ and variance $b''(a)+O(P)$. Consequently, the averaged cross term in (\ref{eq:conditional-risk}) contributes only to the $O(P^5)$ remainder in the risk difference, while the average of $\norm{u}^2$ gives the positive leading coefficient $\{b''(a)\}^3$.

At a fixed covariance, higher-order terms can change the ranking. A sufficient condition for ISD to have the smaller posterior risk is that the distance from the implicit mode to the posterior mean is less than half the distance between the explicit and implicit updates. Appendix D bounds this distance using posterior concentration and a global Lipschitz bound on the Hessian of the log likelihood. For Poisson log intensity, the result is sharper. The posterior mean of the log intensity lies below the implicit mode and within $P/2$ of it. The sharper upper bound on the mean--mode gap, obtained from the root equations in Appendix D, has the same leading term $\tfrac{1}{2}e^aP^2$ as the gap itself. Hence, if the count exceeds the intensity at the predictive mean, ISD has the smaller posterior risk for every $P>0$. The bounds follow from two scalar root equations and require no posterior integration; Section~\ref{sec:example-poisson} explains the ordering.

\subsection{Predictive covariance and Fisher scaling}\label{sec:fisher}

So far, $P_t$ has been given. We now ask when $P_t\propto\I_t(a_t)^{-1}$, which yields the inverse-Fisher scaling used in score-driven models. Proposition \ref{prop:information_matched_covariance} shows that this proportionality is equivalent to a local precision-discounting condition.

\begin{proposition}
\label{prop:information_matched_covariance}
Suppose that the likelihood is replaced locally by its
Fisher-scoring quadratic surrogate around \(a_t\),
$\tilde \ell_t(\theta)
=
\ell_t(a_t)
+
s_t(a_t)^\prime(\theta-a_t)
-
\tfrac{1}{2}(\theta-a_t)^\prime\I_t(a_t)(\theta-a_t)$,
where \(\I_t(a_t)\) is positive definite. Under the Gaussian predictive law $\N(a_t,P_t)$, the surrogate posterior is Gaussian with covariance $W_t$ satisfying
$W_t^{-1}
=
P_t^{-1}+\I_t(a_t)$
and posterior-mean correction $W_t s_t(a_t)$.
If the local predictive precision is obtained by discounting
the filtered precision, $P_t^{-1}
=
\delta W_t^{-1}$ with $0<\delta<1$, then
\[
P_t=\kappa_P\I_t(a_t)^{-1},\qquad W_t=(1-\delta)\I_t(a_t)^{-1},\qquad \kappa_P=\frac{1-\delta}{\delta}.
\]
If, in addition, Assumption \ref{ass:local} holds with $K=0$ and \(\|\I_t(a_t)^{-1}\|\) is bounded as \(\kappa_P\to0\), then
$
\E(\theta_t| y_t,\F_{t-1})
=
a_t+
\kappa_P\I_t(a_t)^{-1}s_t(a_t)
+
O(\kappa_P^2).
$
\end{proposition}

The discounting condition applies within one update and does not propagate a covariance across dates. If $P_t$ is specified as a function of $a_t$, the derivative in (\ref{eq:gaussian-tweedie}) is still taken with $P_t$ held fixed.

\begin{remark}
The predictive covariance scale $\kappa_P=(1-\delta)/\delta$ in Proposition \ref{prop:information_matched_covariance} differs from the gain $1-\delta$ in Proposition \ref{prop:exact_nef_conjugate}. The first scales the predictive covariance, $P_t=\kappa_P\I_t(a_t)^{-1}$, and the second scales the filtered covariance, $W_t=(1-\delta)\I_t(a_t)^{-1}$. The gain of the exact conjugate update is the second. The two are related by $1-\delta=\kappa_P/(1+\kappa_P)$ and agree to first order as $\delta\to1$.
\end{remark}

Related local approximations appear in approximate non-Gaussian filtering; see, e.g., \citet{West:1981}, \citet{Shephard:1994}, \citet{Fahrmeir:1992}, and \citet{FahrmeirTutz:2001}. They also appear in the perturbation filters of \citet{BuccheriBormettiCorsiLillo:2023} and in the connection between unobserved-component models and score-driven filters in \citet{CataniaDInnocenzoLuati:2026}. Inverse-Fisher-scaled score updates are natural-gradient steps, and \citet{Ollivier:2018} shows that online natural-gradient learning is equivalent to extended Kalman filtering; see also the Bayesian learning rule of \citet{KhanRue:2023}. The expansions above apply to any differentiable parameter of the conditional density, not only to location, mean, or variance parameters.

\subsection{Propagation of approximation error}\label{sec:propagation}

The results above concern a single update. An observation-driven recursion combines the update with a prediction step, for example $a_{t+1}=\omega+B(a_t+\Delta_t^j)$ for $j\in\{E,I,I\ast\}$, and one-step errors then accumulate at a rate that depends on how strongly the update responds to the incoming mean. For a scalar parameter with $P$ held fixed, the exact posterior mean and the implicit mode satisfy
\begin{equation}\label{eq:scalar-sensitivity}
 \frac{\partial m}{\partial a}=\frac{V}{P},
 \qquad
 \frac{\partial\widehat\theta^I}{\partial a}=\frac{W^I}{P},
\end{equation}
where $V$ is the posterior variance in (\ref{eq:gaussian-cov-identity}) and $W^I$ is the inverse posterior curvature at the mode. Both ratios are at most one when the log likelihood is concave, so the update damps errors in the incoming mean. Appendix E gives the multivariate identities and the proof of Corollary~\ref{cor:recursive} below.

To compare recursions, consider Gaussian predictive laws with covariance $P_t=\varepsilon\Omega_t$. Let $F_t^B(a)$ be the exact posterior mean under $\N(a,P_t)$, let $F_t^j(a)=a+\Delta_t^j(a)$ be update $j$, and let $g_t$ be the prediction map, for example $g_t(m)=\omega+Bm$. The exact and approximate mean recursions are
$$
a_{t+1}^B=g_t\{F_t^B(a_t^B)\},
\qquad
a_{t+1}^j=g_t\{F_t^j(a_t^j)\},
$$
and $e_t^j=(\E\norm{a_t^B-a_t^j}^2)^{1/2}$ measures their distance. Suppose that, on a set that contains all iterates, the exact map contracts, $\norm{g_t\{F_t^B(a)\}-g_t\{F_t^B(b)\}}\leq q_t\norm{a-b}$, and the one-period error is bounded, $\norm{g_t\{F_t^B(a)\}-g_t\{F_t^j(a)\}}\leq\eta_t^j$, where the $\F_t$-measurable bounds satisfy $\E(q_t^2|\F_{t-1})\leq\bar q_\varepsilon^2<1$ and $\E\{(\eta_t^j)^2|\F_{t-1}\}\leq K_j^2\varepsilon^{2(r_j+1)}$ uniformly in $t$.

\begin{corollary}\label{cor:recursive}
Under the conditions above and with $a_1^B=a_1^j$,
$$
\sup_{t\geq1}e_t^j\leq\frac{K_j\varepsilon^{r_j+1}}{1-\bar q_\varepsilon}.
$$
If $1-\bar q_\varepsilon$ is bounded away from zero, the recursive error has the one-step order $\varepsilon^{r_j+1}$. If $1-\bar q_\varepsilon\asymp\varepsilon^\kappa$ with $\kappa>0$, it is $O(\varepsilon^{r_j+1-\kappa})$.
\end{corollary}

Theorem~\ref{thm:comparison} and Proposition~\ref{prop:mcisd} give $r_E=r_I=1$ and $r_{I\ast}=2$, provided the local errors satisfy the moment bound. The contraction coefficient combines persistence with the information in the observation. If $\Omega_t=\Omega$, $g_t$ has Lipschitz coefficient $\rho$ in the $\Omega^{-1}$ norm, and $-\nabla^2\ell_t\succeq j_0\Omega^{-1}$ with $j_0>0$, then the multivariate version of (\ref{eq:scalar-sensitivity}) gives $\bar q_\varepsilon=\rho/(1+\varepsilon j_0)$. A fixed $\rho<1$ preserves the one-step orders. For a random walk, $\rho=1$, informative observations still supply contraction, but the gap is of order $\varepsilon$, and the bounds become $O(\varepsilon)$ for ESD and ISD and $O(\varepsilon^2)$ for \ISDp{}. The factor $(1-\bar q_\varepsilon)^{-1}$ cannot be removed in general. \citet{DInnocenzo:2026} develops related approximation bounds for score-driven filters, studying initialization effects and empirical processes of reconstructed innovations. Here the target of Corollary~\ref{cor:recursive} is the exact Gaussian-prior mean recursion, and the error arises from approximating its posterior-mean updates. A Gaussian predictive approximation adds a further error when the true predictive distribution is not Gaussian, and the recursive experiment in Section~\ref{sec:numerics-recursive} measures this error separately.

\section{Examples}\label{sec:examples}

This section revisits models for which much is already known. The Kalman filter, GARCH, and the Poisson--gamma filter of \citet{HarveyFernandes:1989} are standard, and \citet{LangeVanOsVanDijk:2026} show that ISD reproduces the Kalman update in the Gaussian location model. What we add is the Bayesian interpretation of these updates and the comparison of ESD, ISD, and \ISDp{} in each model. For the Gaussian variance, Poisson, and Gamma models, Proposition \ref{prop:exact_nef_conjugate} gives exact posterior-mean updates in expectation space. Adding a transition equation to these updates gives familiar observation-driven recursions, which are not posterior means. With a Gaussian predictive law for the log variance or the log intensity, the same models show that ISD captures the curvature and leaves mean--mode errors of opposite signs. The Bernoulli log-odds model shows that ESD can have the lower conditional risk while ISD has the lower average risk. In each comparison, the updates use the same predictive mean and covariance and target the posterior mean of the parameter that has the Gaussian predictive law.

\subsection{Gaussian location}

Consider $Y_t|\mu_t\sim \N(\mu_t,\sigma^2)$ with $\mu_t|\F_{t-1}\sim \N(\mu_{t|t-1},P_{t|t-1})$.
The predictive marginal density is Gaussian with mean $\mu_{t|t-1}$ and variance $P_{t|t-1}+\sigma^2$, so Tweedie's formula gives the Kalman update
$$
\E[\mu_t|Y_t=y_t,\F_{t-1}]
=
\mu_{t|t-1}
+
\frac{P_{t|t-1}}{P_{t|t-1}+\sigma^2}
(y_t-\mu_{t|t-1}).
$$
The conditional likelihood score at the predictive mean is
$s_t(\mu_{t|t-1})=(y_t-\mu_{t|t-1})/\sigma^2$ with $\I_\mu=1/\sigma^2$, so that $\I_\mu^{-1}s_t(\mu_{t|t-1})=y_t-\mu_{t|t-1}$ and the Kalman correction is an inverse-Fisher-scaled conditional-score correction. If $P_{t|t-1}=\sigma^2/n_{t|t-1}$, the Kalman gain is $1/(n_{t|t-1}+1)$, and under the local precision-discounting condition
\(n_{t|t-1}=\delta/(1-\delta)\) it equals $1-\delta$. Thus, under precision discounting, the Tweedie, Kalman, and inverse-Fisher-scaled score updates coincide.

For the comparison of ESD and ISD with a common $P_t=P_{t|t-1}$, write $q=P_t/\sigma^2$ and $r_t=y_t-\mu_{t|t-1}$. Then
\begin{equation}\label{eq:gaussian-comparison}
 \Delta_t^B=\Delta_t^I=\frac{q}{1+q}r_t,\qquad
 \Delta_t^E=qr_t,\qquad \Delta_t^{(2)}=q(1-q)r_t.
\end{equation}
ISD retains the Kalman denominator $1+q$ for every $q>0$ \citep{LangeVanOsVanDijk:2026}, whereas ESD corresponds to its first-order expansion. The exact correction lies between zero and the innovation, but ESD moves beyond the observation when $q>1$, and the explicit second-order correction then has the wrong sign. Thus, a higher-order expansion need not be more accurate at finite covariance. The posterior is Gaussian, so the \ISDp{} adjustment is zero, and (\ref{eq:scalar-sensitivity}) gives the sensitivity $\sigma^2/(P_t+\sigma^2)$ to the incoming mean. An explicit update with the gain $q/(1+q)$ is also exact, as in Proposition \ref{prop:exact_nef_conjugate}, whereas (\ref{eq:gaussian-comparison}) compares the updates with the common covariance $P_t$.

\subsection{Gaussian variance and GARCH}\label{sec:example-variance}

Consider the zero-mean Gaussian variance model
$Y_t|h_t\sim \N(0,h_t)$.
The conditional log likelihood satisfies
$$
\partial_h\ell_t(h_t)=\frac{y_t^2-h_t}{2h_t^2},
\qquad
\I_h(h_t)=\frac{1}{2h_t^2}.
$$
Hence the inverse-Fisher-scaled score is
$\I_h(h_t)^{-1}\partial_h\ell_t(h_t)=y_t^2-h_t$.
The sufficient statistic for $h_t$ is $X_t=Y_t^2$, which, conditional on $h_t$, has a Gamma distribution with shape $1/2$, mean $h_t$, and variance $2h_t^2$, so $X_t$ is a natural-exponential-family observation with expectation parameter $h_t$. We apply Proposition \ref{prop:exact_nef_conjugate} with $X_t$ as the observation. The conjugate prior then implies an inverse-gamma distribution for $h_t$, and Proposition \ref{prop:exact_nef_conjugate} gives the martingale update
\begin{equation}\label{eq:variance-conjugate}
h_{t|t}
=
h_{t|t-1}
+
(1-\delta)(y_t^2-h_{t|t-1}).
\end{equation}

A standard observation-driven variance recursion embeds the same innovation in a transition equation,
\begin{equation}\label{eq:garch}
h_{t+1}
=
\omega+\beta h_t+\alpha(y_t^2-h_t)
=
\omega+\alpha y_t^2+(\beta-\alpha)h_t.
\end{equation}
This is the GARCH(1,1) recursion of \citet{Bollerslev:1986}, reparameterized with $\alpha_G=\alpha$ and $\beta_G=\beta-\alpha$, so the usual nonnegativity conditions are $\alpha\geq0$, $\beta\geq\alpha$, and $\omega>0$. The update (\ref{eq:variance-conjugate}) is the exact conjugate posterior mean, whereas the GARCH recursion (\ref{eq:garch}) adds autoregressive dynamics and, without further assumptions, is not a posterior mean. For $0<\alpha<\beta$, writing (\ref{eq:garch}) as $h_{t+1}=\omega+\beta\{h_t+(\alpha/\beta)(y_t^2-h_t)\}$ shows that the GARCH recursion applies the conjugate update (\ref{eq:variance-conjugate}) to $h_t$ with gain $\alpha/\beta\in(0,1)$ and then the affine prediction step $h\mapsto\omega+\beta h$. The boundary cases $\alpha=0<\beta$ and $\alpha=\beta>0$ correspond to limiting gains zero and one. If $\beta=0$, nonnegativity forces $\alpha=0$ and the recursion reduces to $h_{t+1}=\omega$. In the interior case, the GARCH coefficient $\alpha$ therefore corresponds to the discount factor $\delta=1-\alpha/\beta$, which is close to $1-\alpha$ when the persistence $\beta$ is close to one. For example, $\alpha=0.05$ and $\beta=0.98$ give $\delta\approx0.949$ and an effective memory of $(1-\delta)^{-1}=\beta/\alpha\approx20$ observations. \citet{DonkervanHeelShephard:2026} obtain a GARCH(1,1) recursion of the form (\ref{eq:garch}) from an exponentially weighted log-likelihood that is anchored at the unconditional mean of $Y_t^2$, and (\ref{eq:variance-conjugate}) without the anchor.

For the log variance, $\eta_t=\log h_t$, we have
$\partial_\eta\ell_t=\tfrac{1}{2}(y_t^2/h_t-1)$ and $\I_\eta=\tfrac{1}{2}$,
so inverse-Fisher scaling gives the relative innovation $y_t^2/h_t-1$, which has conditional variance $2$, whereas $y_t^2-h_t$ has conditional variance $2h_t^2$. Score-driven models for the log scale under generalized-$t$ observation distributions are developed by \citet{HarveyLange:2017}. The parameterization also changes the Bayesian target, because $\E(h_t| y_t)\neq\exp\{\E(\log h_t| y_t)\}$ in general. To compare ESD and ISD, let $\eta_t|\F_{t-1}\sim\N(a_t,P_t)$ with $P_t>0$, so that $h_t$ has a lognormal predictive distribution instead of the inverse-gamma distribution behind (\ref{eq:variance-conjugate}). Both updates then approximate the posterior mean of the log variance, with the corrections
\begin{equation}\label{eq:logvariance-updates}
 \Delta_t^E=\frac{P_t}{2}\bigl(y_t^2e^{-a_t}-1\bigr),\qquad
 \Delta_t^I=\frac{P_t}{2}\bigl(y_t^2e^{-(a_t+\Delta_t^I)}-1\bigr).
\end{equation}
The log posterior is strictly concave for every $P_t>0$, so the implicit equation has a unique solution, the posterior mode $\widehat\eta_t^I=a_t+\Delta_t^I$. For $y_t\ne0$, the score is decreasing in the log variance, so the two corrections have the same sign and the implicit correction lies between zero and the explicit correction. An upward revision increases the fitted variance and reduces the score, so ISD responds less than ESD to a large squared observation. If $y_t=0$, both corrections equal $-P_t/2$. This comparison holds $P_t$ fixed and isolates the effect of evaluating the score implicitly. Estimating the gain may partly offset the attenuation, but the size and direction of that adjustment depend on the model and the data.

The third derivative of the log likelihood is positive for $y_t\ne0$, so, by Theorem~\ref{thm:comparison}, the posterior mean of $\eta_t$ lies above the mode for small $P_t$. With $W_t^I=\{P_t^{-1}+y_t^2e^{-\widehat\eta_t^I}/2\}^{-1}$, the \ISDp{} update (\ref{eq:mcisd}) is
\[
 \Delta_t^{I\ast}=\Delta_t^I+
       \frac{(W_t^I)^2}{4}y_t^2e^{-\widehat\eta_t^I}.
\]
The added term is positive unless $y_t=0$, and it removes the leading mean--mode error. \citet{WuHe:2026} study quasi score-driven volatility models in the different asymptotic regime where the observation interval tends to zero.

\subsection{Poisson and Gamma observation models}\label{sec:example-poisson}

For $Y_t|\mu_t\sim\operatorname{Poisson}(\mu_t)$, the variance function is $V(\mu_t)=\mu_t$, and Lemma \ref{lem:variance_function_score} gives
$$
\partial_\mu\ell_t(\mu_t)=\frac{y_t-\mu_t}{\mu_t},
\qquad
\I_\mu(\mu_t)=\frac{1}{\mu_t},
\qquad
\I_\mu(\mu_t)^{-1}\partial_\mu\ell_t(\mu_t)=y_t-\mu_t.
$$
A conjugate prior on $\log\mu_t$ is a gamma prior on $\mu_t$, and, under precision discounting, Proposition \ref{prop:exact_nef_conjugate} gives
$$
\mu_{t|t}
=
\mu_{t|t-1}
+
(1-\delta)(y_t-\mu_{t|t-1}).
$$
This recursion is the discounted Poisson--gamma filter of \citet{HarveyFernandes:1989}, who derived it from a gamma prior with discounted parameters and showed that its one-step forecasts are exponentially weighted moving averages. Proposition \ref{prop:exact_nef_conjugate} adds that the innovation $y_t-\mu_{t|t-1}$ is the inverse-Fisher-scaled conditional score. The discreteness of $Y_t$ affects only the marginal Tweedie identity in Section~\ref{sec:tweedie}, not the conditional score. Related conditional-mean recursions for counts appear in exponential-family INGARCH models; \citet{HuangFungMacaskillAukes:2026} study semiparametric estimation for this class.

For a Gamma observation model with fixed dispersion $\phi$ and time-varying mean $\mu_t$, the variance function is $V(\mu_t)=\mu_t^2$, so Lemma \ref{lem:variance_function_score} gives
$\partial_\mu\ell_t(\mu_t)=(y_t-\mu_t)/(\phi\mu_t^2)$, $\I_\mu(\mu_t)=1/(\phi\mu_t^2)$, and again $\I_\mu(\mu_t)^{-1}\partial_\mu\ell_t(\mu_t)=y_t-\mu_t$.
For fixed \(\phi\), this is a one-parameter natural exponential family after absorbing
\(\phi\) into the canonical parameter and log-partition function, and the exact conjugate-prior posterior-mean update in expectation space has the same form,
$\mu_{t|t}=\mu_{t|t-1}
+(1-\delta)(y_t-\mu_{t|t-1})$, with its own prior strength.

For the Gaussian comparison in the Poisson model, write $\eta_t=\log\mu_t$ and let $\eta_t|\F_{t-1}\sim\N(a_t,P_t)$. The resulting prior on intensity is lognormal, and both updates target the posterior mean of log intensity. The score is $y_t-e^{\eta_t}$, and the two corrections are
\begin{equation}\label{eq:poisson-updates}
 \Delta_t^E=P_t\bigl(y_t-e^{a_t}\bigr),\qquad
 \Delta_t^I=P_t\bigl(y_t-e^{a_t+\Delta_t^I}\bigr).
\end{equation}
The log posterior has strictly negative second derivative $-P_t^{-1}-e^{\eta_t}$, so the implicit solution is unique for every $P_t>0$ and gives the posterior mode of log intensity. As in the log-variance example, the implicit correction has the same sign as the explicit correction and smaller magnitude whenever the score at $a_t$ is nonzero. For an unexpectedly large count, the intensity rises as the parameter is revised, which reduces the remaining count innovation within the implicit equation.

Here the posterior mean lies below the mode, the opposite of the log-variance case; see Appendix D. If $y_t>e^{a_t}$, ESD lies above the mode, so ISD is closer to the posterior mean for every $P_t>0$. For smaller counts, the attenuation of ISD does not determine the ranking, but Corollary~\ref{cor:average-risk} gives ISD the lower average risk for small $P_t$. With $\widehat\eta_t^I=a_t+\Delta_t^I$ and $W_t^I=\{P_t^{-1}+e^{\widehat\eta_t^I}\}^{-1}$, the \ISDp{} update is
\[
 \Delta_t^{I\ast}=\Delta_t^I-
       \frac{(W_t^I)^2}{2}e^{\widehat\eta_t^I}.
\]
It moves the implicit estimate down, toward the posterior mean, and Proposition~\ref{prop:mcisd} gives its higher accuracy for small $P_t$. In both the log-variance and the log-intensity models, the implicit correction solves a scalar equation on the interval between zero and the explicit correction, and the mean correction moves the estimate up for the log variance and down for the log intensity.

\subsection{Bernoulli log odds and conditional risk reversal}
\label{sec:example-bernoulli}

Let $Y|\eta\sim\operatorname{Bernoulli}\{p(\eta)\}$, where $p(\eta)=e^\eta/(1+e^\eta)$, and let $\eta\sim\N(a,P)$. The target is the posterior mean of log odds. The score is $s(\eta)=y-p(\eta)$, giving
\[
 \Delta^E=P\{y-p(a)\},\qquad
 \Delta^I=P\{y-p(a+\Delta^I)\}.
\]
The log posterior has second derivative $-P^{-1}-p(\eta)\{1-p(\eta)\}<0$, so the implicit solution is the unique posterior mode. Its correction lies between zero and the explicit correction, but this attenuation need not improve posterior-mean accuracy.

Take $a=\log4$, so $p(a)=0.8$, and put $P=\varepsilon$. The curvature and asymmetry terms in Theorem~\ref{thm:comparison} are $u=-0.16(y-0.8)$ and $v=0.048$. Table~\ref{tab:bernoulli-reversal} gives the coefficient of $\varepsilon^4$ in the ESD-minus-ISD posterior risk difference (\ref{eq:conditional-risk}).

\begin{table}[!htbp]
\tableformat
\caption{Conditional risk reversal for Bernoulli log odds}
\label{tab:bernoulli-reversal}
\begin{papertabular}{lrrrl}
\toprule
Outcome & $u$ & $v$ & $u^2+2uv$ & Lower local risk \\
\midrule
Failure ($y=0$) & $0.128$ & $0.048$ & $0.028672$ & ISD \\
Success ($y=1$) & $-0.032$ & $0.048$ & $-0.002048$ & ESD \\
\bottomrule
\end{papertabular}
\begin{floatnotes}
Notes: The comparison uses a Gaussian prior on log odds with mean
$\log4$, variance $\varepsilon$, and squared-error loss. The rankings hold
for sufficiently small positive $\varepsilon$.
\end{floatnotes}
\end{table}

After a success, the curvature and asymmetry terms partially cancel. ESD's leading error involves $u+v=0.016$, whereas ISD's involves $v=0.048$. After a failure, the two terms reinforce each other and ISD is more accurate. The prior-predictive success probability tends to $0.8$ as $\varepsilon\downarrow0$, so the average-risk coefficient is
\[
 0.2(0.028672)+0.8(-0.002048)=0.004096=\{b''(a)\}^3>0 .
\]
The average-risk comparison therefore favors ISD locally, as in Corollary~\ref{cor:average-risk}, even though ESD has lower conditional risk after the more probable outcome. Section~\ref{sec:numerics-risk} examines this reversal at finite predictive variance.

\section{Numerical Evidence}\label{sec:numerics}

The theory in Section~\ref{sec:local} is local in the predictive covariance. We examine whether its conclusions hold at finite predictive covariance for the accuracy of a single update (Section~\ref{sec:numerics-static}), for conditional and average risks (Section~\ref{sec:numerics-risk}), and when updates are repeated in a recursion whose predictive distribution is not Gaussian (Section~\ref{sec:numerics-recursive}). All comparisons use the same predictive covariance for ESD, ISD, \ISDp{}, and the explicit second-order update (\ref{eq:second-order-filter}), isolating differences between the update formulas rather than differences in gain estimation. All results are computed with our Julia package TweedieScore.jl \citep{HansenTong:2026Julia}, which also reproduces every table and figure. Implementation details, numerical checks, and additional tables are in Appendix F.

\subsection{Approximation under an exact Gaussian predictive law}
\label{sec:numerics-static}

In each replication, we draw a scalar $\theta\sim\N(a,P)$ and one observation from its conditional density, and we compare the four updates with the posterior mean. Because the predictive law is set anew in every replication, the experiment involves neither parameter estimation nor recursion. The designs are Poisson log intensity, Gaussian log variance, and Student-$t_5$ location. The first two use the corrections in (\ref{eq:poisson-updates}) and (\ref{eq:logvariance-updates}), so the experiment measures the curvature and mean--mode effects of Sections~\ref{sec:example-variance} and \ref{sec:example-poisson}. The predictive mean is $\log4$ for Poisson and zero for the other models. Student-$t$ observation errors are standardized to unit variance. At each of twelve predictive variances between $0.00015625$ and $0.32$, all methods use 100,000 common prior-predictive draws and the same predictive covariance. The posterior means used as reference values are evaluated by Gauss--Hermite quadrature.

\begin{table}[t]
\tableformat
\caption{RMSE of posterior-mean approximations}
\label{tab:simulation-summary}
\begin{papertabular}{llrrrr}
\toprule
Model & $P$ & ESD & ISD & \ISDp{} & Explicit 2nd \\
\midrule
Poisson log intensity & 0.01 & $8.224\times10^{-4}$ & $1.854\times10^{-4}$ & $3.568\times10^{-7}$ & $2.733\times10^{-5}$ \\
 & 0.04 & $1.291\times10^{-2}$ & $2.395\times10^{-3}$ & $7.252\times10^{-6}$ & $1.689\times10^{-3}$ \\
 & 0.16 & $1.986\times10^{-1}$ & $1.873\times10^{-2}$ & $1.386\times10^{-4}$ & $1.008\times10^{-1}$ \\
 & 0.32 & $7.884\times10^{-1}$ & $3.737\times10^{-2}$ & $5.193\times10^{-4}$ & $7.986\times10^{-1}$ \\
Gaussian log variance & 0.01 & $1.895\times10^{-4}$ & $4.111\times10^{-5}$ & $8.887\times10^{-8}$ & $1.416\times10^{-5}$ \\
 & 0.04 & $3.021\times10^{-3}$ & $5.740\times10^{-4}$ & $3.581\times10^{-6}$ & $1.037\times10^{-3}$ \\
 & 0.16 & $4.109\times10^{-2}$ & $6.266\times10^{-3}$ & $8.187\times10^{-5}$ & $5.211\times10^{-2}$ \\
 & 0.32 & $1.691\times10^{-1}$ & $1.814\times10^{-2}$ & $3.426\times10^{-4}$ & $6.709\times10^{-1}$ \\
Student-$t_{5}$ location & 0.01 & $1.602\times10^{-4}$ & $5.796\times10^{-5}$ & $5.601\times10^{-7}$ & $3.560\times10^{-6}$ \\
 & 0.04 & $2.446\times10^{-3}$ & $8.452\times10^{-4}$ & $2.923\times10^{-5}$ & $2.111\times10^{-4}$ \\
 & 0.16 & $3.320\times10^{-2}$ & $9.826\times10^{-3}$ & $9.486\times10^{-4}$ & $1.044\times10^{-2}$ \\
 & 0.32 & $1.117\times10^{-1}$ & $2.801\times10^{-2}$ & $3.778\times10^{-3}$ & $6.460\times10^{-2}$ \\
\midrule
\multicolumn{6}{@{}l@{}}{\emph{Estimated small-$P$ convergence slopes}} \\
Poisson log intensity & & 1.999 & 1.994 & 2.973 & 2.999 \\
Gaussian log variance & & 1.996 & 1.996 & 2.984 & 2.940 \\
Student-$t_{5}$ location & & 1.999 & 1.998 & 2.995 & 2.998 \\
\bottomrule
\end{papertabular}
\begin{floatnotes}
Notes: Entries above the midrule are Monte Carlo RMSEs based on 100,000
prior-predictive draws. Exact posterior means use 80-node Gauss--Hermite
quadrature. Slopes are coefficients from log-RMSE regressions on log $P$ using
$P\leq0.0025$ (five covariance values). The local expansions suggest slopes
two for ESD and ISD and three for \ISDp{} and explicit second order, provided
the leading coefficients are nonzero and square-integrable remainder
bounds permit averaging; see Appendix F.
\end{floatnotes}
\end{table}

\begin{figure}[t]
\centering
\includegraphics[width=\textwidth]{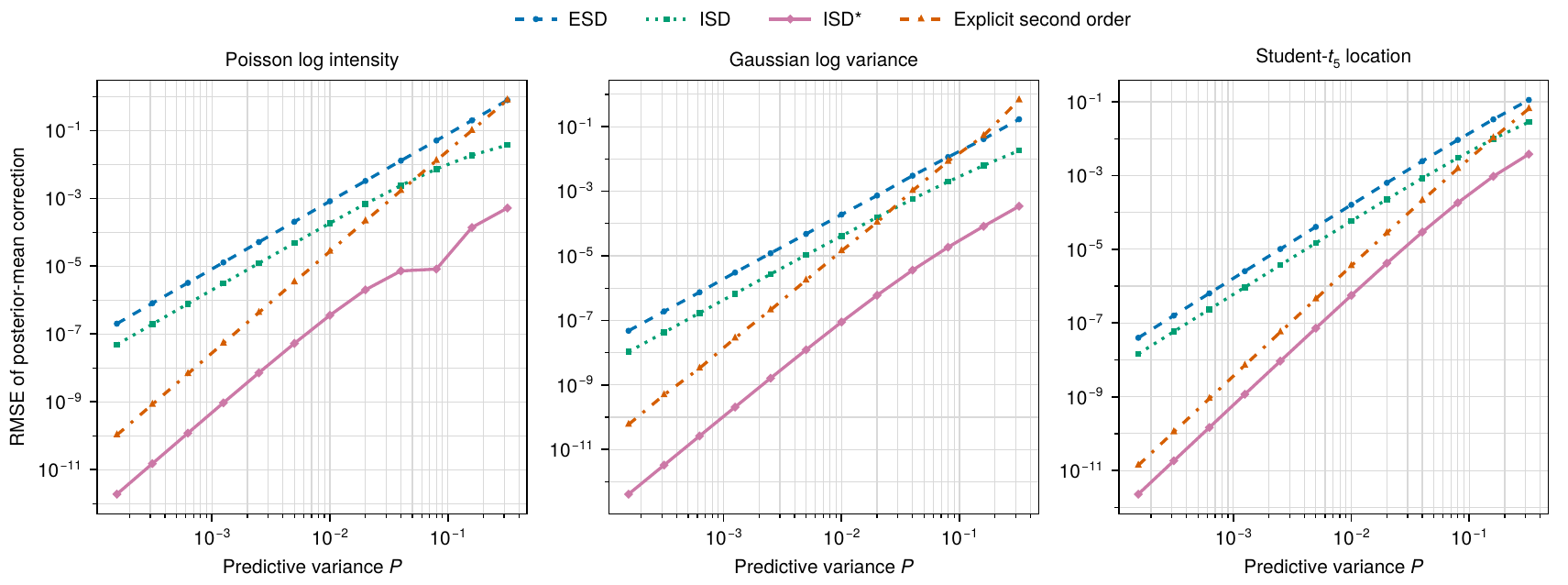}
\caption{RMSE of posterior-mean approximations across predictive variances}
\label{fig:approximation-error}
\begin{floatnotes}
Notes: Points are Monte Carlo RMSEs relative to the Gaussian-prior
posterior mean, evaluated by 80-node Gauss--Hermite quadrature.
Every cell uses 100,000 common draws from the same simulation run as
Table~\ref{tab:simulation-summary}. Both axes are logarithmic.
Explicit second order is defined in (\ref{eq:second-order-filter}).
\end{floatnotes}
\end{figure}

In Table~\ref{tab:simulation-summary}, the estimated slopes are close to two for ESD and ISD and close to three for \ISDp{} and the explicit second-order update, as the local expansions predict. ISD nevertheless has a lower RMSE than ESD in all 36 model--variance cells, and at $P=0.32$ the reduction ranges from about $75\%$ for Student-$t_5$ location to $95\%$ for Poisson log intensity. Thus, the error order does not measure the gain from the curvature adjustment. \ISDp{} has the lowest RMSE in every cell. The explicit second-order update has the same local order as \ISDp{}, but it can break down at finite covariance. At $P=0.32$, its RMSE for log variance is $0.671$, compared with $0.169$ for ESD and $0.000343$ for \ISDp{}, and Figure~\ref{fig:approximation-error} shows that the two third-order updates diverge as the predictive variance increases. The slopes describe the small-variance grid. To interpret them as orders of the integrated errors, the expansion remainder must be controlled over observations, as it is under the conditions of Corollary~\ref{cor:average-risk}. For $\nu=3$ and $10$, the slopes are essentially unchanged, but the accuracy at finite variance depends on the tail thickness; see Appendix F.

\subsection{Conditional and average Bayesian risk}
\label{sec:numerics-risk}

Appendix F reports deterministic calculations of the risk after each observation and of its average over observations, for Bernoulli log odds and Poisson log intensity with $a=\log4$ and $P=\varepsilon$, and for a bivariate logistic design with interacting parameters. For Bernoulli log odds at $P=0.01$, the ESD-minus-ISD posterior risk is $2.865\times10^{-10}$ after $y=0$ and $-2.029\times10^{-11}$ after $y=1$, so ESD has the lower risk after a success, whereas ISD has the lower average risk, as in Section~\ref{sec:example-bernoulli}. Normalized by $\varepsilon^4$, the average risk difference converges to the coefficient in Corollary~\ref{cor:average-risk} and its multivariate version, and at $\varepsilon=10^{-3}$ the relative discrepancy is below $0.07\%$ in every model. In the bivariate design at $\varepsilon=0.32$, the RMSEs of ESD, ISD, and \ISDp{} are $0.03913$, $0.008698$, and $0.001347$. Without computing the posterior mean, the Poisson enclosure of Appendix D shows that ISD has the lower risk for observations with about $81\%$ to $89\%$ of the predictive probability, depending on the variance.

\subsection{A recursive non-Gaussian benchmark}
\label{sec:numerics-recursive}

To study how approximation errors accumulate, we use a Gaussian AR(1) latent-location model with independent, unit-variance Student-$t$ observation errors. Persistence is $0.80$ or $0.98$, the stationary state variance is $0.25$ or $1$, and the degrees of freedom are $3$, $5$, or $10$. Each of the twelve designs uses 1,000 paths of 500 periods, and a deterministic grid filter gives the state-space reference \citep{Kitagawa:1987,DurbinKoopman:2012}. \citet{KoopmanLucasScharth:2016} compare score-driven and parameter-driven filters in related designs. Here we decompose the difference into an approximation error and a predictive-shape error.

The Student-$t$ update makes the filtered distribution non-Gaussian, so the next predictive distribution is generally not Gaussian, although the state transition is Gaussian. For the one-step comparisons, we therefore also compute the exact posterior mean under a Gaussian prior with the predictive mean and variance of the grid filter. This separates the error in approximating the Gaussian-prior correction from the error caused by replacing the predictive distribution with a Gaussian one. Every recursion propagates its own mean and uses the predictive variances of the grid filter. With this common oracle variance, the experiment isolates the mean updates, but it does not evaluate a feasible forecasting procedure. The Gaussian-reset recursion uses the exact Gaussian-prior posterior mean, and the Kalman recursion also replaces the Student-$t$ density by a Gaussian density with the same variance.

\begin{table}[t]
\tableformat
\caption{One-step and accumulated recursive approximation errors}
\label{tab:recursive-compact}
\begin{papertabular}{clrrrr}
\toprule
& & \multicolumn{2}{c}{Oracle one-step RMSE} & \multicolumn{2}{c}{Recursive RMSE} \\
\cmidrule(lr){3-4}\cmidrule(lr){5-6}
$\nu$ & Method & Gaussian target & Grid target & Gaussian reset & Grid filter \\
\midrule
3 & Gaussian reset & 0.000 & 0.006 & 0.000 & 0.007 \\
 & ESD & 0.236 & 0.237 & 0.251 & 0.253 \\
 & ISD & 0.056 & 0.057 & 0.063 & 0.065 \\
 & \ISDp{} & 0.017 & 0.018 & 0.019 & 0.019 \\
 & Explicit 2nd & 0.434 & 0.434 & 0.486 & 0.485 \\
 & Kalman & 0.145 & 0.145 & 0.201 & 0.200 \\
5 & Gaussian reset & 0.000 & 0.003 & 0.000 & 0.004 \\
 & ESD & 0.157 & 0.157 & 0.170 & 0.171 \\
 & ISD & 0.032 & 0.032 & 0.037 & 0.038 \\
 & \ISDp{} & 0.005 & 0.006 & 0.006 & 0.007 \\
 & Explicit 2nd & 0.152 & 0.151 & 0.177 & 0.177 \\
 & Kalman & 0.080 & 0.080 & 0.110 & 0.109 \\
10 & Gaussian reset & 0.000 & 0.002 & 0.000 & 0.002 \\
 & ESD & 0.131 & 0.132 & 0.142 & 0.143 \\
 & ISD & 0.016 & 0.016 & 0.019 & 0.020 \\
 & \ISDp{} & 0.001 & 0.002 & 0.001 & 0.003 \\
 & Explicit 2nd & 0.082 & 0.082 & 0.098 & 0.098 \\
 & Kalman & 0.037 & 0.037 & 0.051 & 0.051 \\
\bottomrule
\end{papertabular}
\begin{floatnotes}
Notes: Entries pool the four persistence--SNR designs for each $\nu$.
One-step Gaussian-target errors evaluate every correction at the grid filter's
predictive mean and variance. Grid-target errors additionally include the
Gaussian-closure component. Recursive columns use each method's propagated mean
and the common oracle variance sequence from the grid filter. Gaussian reset
uses the exact Gaussian-prior correction, so its two Gaussian-target entries are zero.
\end{floatnotes}
\end{table}

Table~\ref{tab:recursive-compact} reports one-step errors, evaluated at the moments of the grid filter, and errors accumulated through the recursion. Under ``Recursive RMSE,'' the column ``Gaussian reset'' gives the distance to the Gaussian-prior mean recursion, which is the target of the bounds in Section~\ref{sec:propagation}, and the column ``Grid filter'' adds the predictive-shape error. In every column and tail-thickness group, the RMSE of ISD is roughly $74\%$ to $88\%$ below that of ESD, so the advantage of ISD persists when the predictive-shape error and the recursive feedback are included. \ISDp{} reduces the error further. Relative to the grid filter, the pooled recursive RMSE falls from $0.0449$ for ISD to $0.0119$ for \ISDp{}. The pooled RMSE of the Gaussian-reset recursion relative to the grid mean is $0.0050$, which measures the cost of the Gaussian predictive approximation. The explicit second-order update has recursive RMSE $0.485$ in the $t_3$ group, so its higher local order does not prevent large errors at the covariances that occur along these paths.

Under squared-error loss, latent-state risk equals the expected posterior variance plus the mean squared distance to the true filtering mean. The common variance term means that improvements in posterior-mean approximation need not yield comparable percentage reductions in latent-state risk. With the stationary variance and the tail thickness held fixed, the ISD error is smaller at persistence $0.98$ than at $0.80$. Because the stationary variance is held fixed, the variance of the state innovation falls as the persistence rises. These results concern the designs studied here and do not verify the uniform assumptions behind the recursive bound. Appendix F reports the full design results, absolute-error measures for $t_3$, and numerical refinements.

\section{Conclusion}\label{sec:conclusion}

In this paper, we have connected score-driven updating with Bayesian updating. The starting point was Tweedie's formula, which expresses a conditional mean in terms of a score. For natural exponential families, it identifies the case where the two updates are identical. The inverse-Fisher-scaled score-driven update of the expectation parameter is the posterior mean under a conjugate prior, and the gain equals $1-\delta$ under the local precision-discounting condition. This gives a Bayesian motivation for the scaling of the score, because conjugacy selects the inverse Fisher information, $d=1$, among the scalings $\I^{-d}$ used in score-driven models. Other scalings and parameterizations can still work better in applications, in particular when the information in the observations varies over time. More generally, under a Gaussian predictive law, the exact Bayesian correction is the predictive covariance multiplied by the score of the predictive marginal density with respect to the predictive mean of the parameter. Replacing that marginal score by the conditional score gives a local approximation with an error of second order in the predictive covariance.

An important insight from connecting score-driven updating with Bayesian updating is that the choice between explicit and implicit updating is a question about the second-order approximation error. ISD captures the curvature term and is exact for quadratic log likelihoods. ISD does not account for the mean--mode term, which is the remaining second-order component. Correcting this term leads to a refined implicit score-driven update, denoted \ISDp{}, that adds the leading Laplace correction from the posterior mode to the posterior mean, and its error is of third order. After a particular observation, either ESD or ISD can be closer to the posterior mean, but ISD has the lower average risk in canonical exponential families when the predictive covariance is small. In the numerical experiments, ISD substantially improves the approximation of the posterior mean at finite predictive covariance, and \ISDp{} improves it further. The recursive experiments distinguish accumulated update errors from the additional error caused by approximating a non-Gaussian predictive distribution with a Gaussian law.

\section*{Acknowledgements}

The authors are grateful to Siem Jan Koopman for helpful discussions and comments that improved the paper. Chen Tong acknowledges financial support from the National Natural Science Foundation of China (72301227, 72673150) and the Fujian Provincial Natural Science Foundation of China (2025J08008).

\section*{Conflict of Interest Statement}

The authors declare no conflicts of interest.

\section*{Data Availability Statement}

Data sharing is not applicable to this article, because no empirical data were analysed. The numerical results are based on simulated data, and the code that reproduces them is available at \url{https://github.com/reinhardhansen/TweedieScore}.

\section*{Declaration of the use of generative AI and AI-assisted technologies}

During the preparation of this work the authors used OpenAI's ChatGPT and
Anthropic's Claude to assist with language editing, LaTeX drafting,
notation checks, and manuscript review. The authors are responsible for
reviewing the resulting text and mathematical arguments and take full
responsibility for the content of the publication.

\small
\bibliographystyle{apalike}
\bibliography{prh}

\normalsize
\clearpage
\appendix
\counterwithin{equation}{section}
\counterwithin{table}{section}
\counterwithin{figure}{section}
\begin{center}\large\bfseries Appendix\end{center}

\noindent Appendix A proves the posterior identities of Sections 2 to 5. Appendix B states the notation for the local expansions and proves the local approximation and comparison results of Section 5. Appendix C proves the average-risk comparison, Appendix D gives comparisons at finite predictive covariance, and Appendix E gives update-sensitivity identities and recursive error bounds, including the proof of Corollary 3. Appendix F documents the numerical implementation and reports additional tables. Results stated only in the appendix carry appendix numbering (Assumption C.1, Theorem C.2, Proposition D.2, and so on).

\section{Proofs for the posterior identities}\label{app:identities}

\subsection{Proof of Lemma 1 (Tweedie's formula)}

Because $\phi_\sigma$ is positive, $f(y)>0$ for every $y$. For all $u$ and $\mu$,
\[
 |\partial_u\phi_\sigma(u-\mu)|
 =\frac{|u-\mu|}{\sigma^2}\phi_\sigma(u-\mu)
 \leq\frac{e^{-1/2}}{(2\pi)^{1/2}\sigma^2},
\]
because $|x|\exp\{-x^2/(2\sigma^2)\}$ is maximized at $|x|=\sigma$. Hence $e^{-1/2}(2\pi)^{-1/2}\sigma^{-2}\pi(\mu)$ is an integrable bound for $|\partial_u\phi_\sigma(u-\mu)\pi(\mu)|$ that does not depend on $u$, and dominated convergence allows differentiation under the integral sign for every prior, so that
\[
 f'(y)=\int-\frac{y-\mu}{\sigma^2}
                  \phi_\sigma(y-\mu)\pi(\mu)d\mu
 =\frac{1}{\sigma^2}\left\{\int\mu\phi_\sigma(y-\mu)\pi(\mu)d\mu-yf(y)\right\}.
\]
The integral $\int\mu\phi_\sigma(y-\mu)\pi(\mu)d\mu$ is finite, because $|\mu|\phi_\sigma(y-\mu)$ is bounded in $\mu$, so the posterior mean exists for every prior. Division by $f(y)$ gives
$\sigma^2 f'(y)/f(y)=\E(\mu| Y=y)-y$.
The noise identity follows from $\varepsilon=y-\mu$ conditional on
$Y=y$. The same argument uses integration with respect to a general
prior measure if no prior density exists. \qed

\subsection{Proof of Lemma 2 (Tweedie's formula for NEFs)}

Since $f(y)/h(y)=\int\exp\{\theta y-\psi(\theta)\}\pi(\theta)d\theta$, which is positive, the stated domination condition justifies differentiating under the integral, and
\[
 \partial_y\log\{f(y)/h(y)\}
 =\frac{\int\theta\exp\{\theta y-\psi(\theta)\}\pi(\theta)d\theta}
        {\int\exp\{\theta y-\psi(\theta)\}\pi(\theta)d\theta}
 =\E(\theta| y).
\]
For the Poisson identity quoted in Section 2,
$\mu p(y|\mu)=(y+1)p(y+1|\mu)$; integration over the prior and
division by $f(y)$ give $\E(\mu| Y=y)=(y+1)f(y+1)/f(y)$. \qed

\subsection{Proof of Lemma 3}

The posterior kernel is $\exp\{\theta y-\psi(\theta)\}\pi(\theta)$. The assumed boundary condition makes it vanish at the endpoints of the natural parameter space, and its derivative is integrable, so
\[
 0=\int\partial_\theta
       [\exp\{\theta y-\psi(\theta)\}\pi(\theta)]d\theta.
\]
After division by the normalizing constant this becomes
$0=y-\E\{\psi'(\theta)| y\}
    +\E\{\partial_\theta\log\pi(\theta)| y\}$, which is the first identity since $\mu=\psi'(\theta)$.
For the conjugate prior $\pi(\theta)\propto\exp\{\tau\theta-n\psi(\theta)\}$, the prior score is $\tau-n\psi'(\theta)$, so $\E(\mu| y)=y+\tau-n\E(\mu| y)$ and
$\E(\mu| y)=(\tau+y)/(n+1)$.
Applied to the prior alone, the same calculation gives
$\E\{\psi'(\theta)\}=\tau/n$. \qed

\subsection{Proof of Lemma 4}

In the mean parameterization of a regular exponential dispersion model, $d\theta/d\mu=1/V(\mu)$ and the natural-parameter score is $(y-\mu)/\phi$. The chain rule gives
$\partial_\mu\ell(y;\mu)=(y-\mu)/\{\phi V(\mu)\}$.
Taking the conditional variance of this score gives
$\I_\mu(\mu)=\operatorname{var}(Y|\mu)/\{\phi^2V(\mu)^2\}=1/\{\phi V(\mu)\}$, and the final identity follows. \qed

\subsection{Proof of Proposition 1 (conjugate updating as a score update)}

The integration-by-parts argument of Lemma 3, applied to the conjugate predictive prior, gives $\mu_{t|t-1}=\E(\mu_t|\F_{t-1})=\tau_{t|t-1}/n_{t|t-1}$.
After observing $y_t$, the likelihood contributes $\exp\{\theta_t y_t-\psi(\theta_t)\}$, so the filtered density is again conjugate, with
$\tau_{t|t}=\tau_{t|t-1}+y_t$ and $n_{t|t}=n_{t|t-1}+1$. The conjugate case of Lemma 3 applied to the filtered density gives
$\mu_{t|t}=(\tau_{t|t-1}+y_t)/(n_{t|t-1}+1)$.
Since $\tau_{t|t-1}=n_{t|t-1}\mu_{t|t-1}$,
\[
\mu_{t|t}
=
\frac{n_{t|t-1}\mu_{t|t-1}+y_t}{n_{t|t-1}+1}
=
\mu_{t|t-1}
+
\frac{y_t-\mu_{t|t-1}}{n_{t|t-1}+1}.
\]
By Lemma 4, $\I_\mu(\mu)^{-1}\partial_\mu\log p_t(y_t|\mu)=y_t-\mu$; evaluating at
$\mu=\mu_{t|t-1}$ gives the stated representation. The local precision-discounting condition gives
$n_{t|t-1}=\delta n_{t|t}=\delta(n_{t|t-1}+1)$, hence
$n_{t|t-1}=\delta/(1-\delta)$ and $(n_{t|t-1}+1)^{-1}=1-\delta$. \qed

\subsection{Proof of Proposition 2}

Since $\nabla_\eta A_t(\eta)=a$, differentiating the predictive marginal density under the integral gives
\[
 \nabla_\eta f_t(y;\eta)
 =\int p_t(y|\theta)(\theta-a)\pi(\theta|\F_{t-1};\eta)d\theta.
\]
Division by $f_t$ expresses the derivative of the log predictive density
as the posterior correction $\E(\theta_t| y,\F_{t-1})-a_t$; evaluating at $y=y_t$ proves the first identity. For the mean parameterization, the Jacobian of the map $\eta\mapsto a$ is the symmetric positive-definite matrix
$P(a)=\nabla^2 A_t(\eta)$, so the chain rule gives
$\nabla_\eta\log f_t=P(a)\nabla_a\log f_t$, proving the second identity. \qed

\subsection{The Gaussian form of Tweedie's formula and the covariance identity}
\label{app:gaussian-form}

Let $\phi_P(\theta-a)$ be the Gaussian predictive density with mean $a$ and fixed covariance $P$. Its derivative in $a$ is
$P^{-1}(\theta-a)\phi_P(\theta-a)$, so
\[
 P\nabla_a f(y;a,P)
 =\int(\theta-a)p(y|\theta)\phi_P(\theta-a)d\theta.
\]
Division by $f$ proves the first equality in equation (2), $m-a=P\nabla_a\log f(y;a,P)$. For the second, use
$\nabla_a\phi_P(\theta-a)=-\nabla_\theta\phi_P(\theta-a)$ and integrate
by parts:
\[
 \nabla_a f(y;a,P)
 =\int\nabla_\theta p(y|\theta)\phi_P(\theta-a)d\theta.
\]
The boundary and integrability conditions justify this integration by
parts. Since $\nabla_\theta p=ps(\theta)$, division by $f$ yields the
posterior score average $\E\{s(\theta)| y\}$. Thus the two equalities are exact; replacing
that average by $s(a)$ is the ESD approximation.

To obtain the covariance identity, equation (3), differentiate the posterior mean
directly. At fixed $P$, the derivative of the normalized posterior density with respect to $a$ is its centered Gaussian score, which gives
$D_a m=\operatorname{var}(\theta| y)P^{-1}$.
Differentiating $m=a+P\nabla_a\log f$ and equating the two
Jacobians yields
$VP^{-1}=I_d+P\nabla_a^2\log f$. Right multiplication by $P$
gives $V=P+P(\nabla_a^2\log f)P$. \qed

\section{Local regularity and approximation proofs}\label{app:comparison}

Throughout this appendix, $y$ is a fixed observation, $L(\theta)=p(y|\theta)$, $\ell=\log L$, and $s=\nabla\ell(a)$, $H=\nabla^2\ell(a)$, and $T=\nabla^3\ell(a)$ are the derivatives of the log likelihood at the predictive mean $a$. For a third-derivative tensor $T$ and symmetric matrix $M$, the
contraction used in Section 5 is the vector with components
$(T:M)_i=\sum_{j,k}T_{ijk}M_{jk}$. Assumption 1, with order $K$, requires $L$ to be positive and $\max(K,3)$ times continuously differentiable on $\R^d$, with
\begin{equation}\label{eq:local-envelope-app}
 \norm{\nabla^k L(a+x)}\leq C L(a)\exp(c\norm{x}^2),\qquad k=0,\ldots,K,\ x\in\R^d,
\end{equation}
for finite constants $C,c$ that do not depend on the predictive covariance. Proposition 3 and Proposition 6 use $K=0$; Theorem 1, Corollary 1, and Proposition 5 use $K=5$. When $P=\varepsilon\Omega$ with $\Omega\succ0$ fixed, we take $\varepsilon$ sufficiently small that
$2c\varepsilon\lambda_{\max}(\Omega)<1$. For the derivative orders covered by the envelope condition, the corresponding likelihood derivatives are integrable under the Gaussian predictive density. When derivative envelopes are imposed, they also justify the corresponding integration-by-parts steps below. The $K=0$ envelope controls likelihood tails; the proof of Proposition 3 instead expands the posterior numerator and denominator directly. The condition is pointwise in $y$; the integrated comparison requires the additional assumption in Appendix~\ref{app:average-risk}.

\subsection{Proof of Proposition 3}

The proof uses the $k=0$ envelope in (\ref{eq:local-envelope-app}) together with three continuous derivatives of $\ell$ on a ball $\{\theta:\norm{\theta-a}\leq r\}$, on which the third-order partial derivatives are bounded by a constant $L_3$. Once $2c\lambda_{\max}(P)<1$, the $k=0$ envelope makes $L(\theta)$ and $\norm{\theta-a}L(\theta)$ integrable under $\N(a^\prime,P)$, uniformly for $a^\prime$ near $a$. Hence the conditions of Proposition 2 hold, and
$\E(\theta| y)-a=P\nabla_a\log f(y;a,P)$.
Set $u=\theta-a$, so that $u\sim\N(0,P)$ under the predictive law, and normalize $\ell(a)=0$; this rescales numerator and denominator below by the same factor and leaves their ratio unchanged. Bayes' rule gives
\[
P\nabla_a\log f(y;a,P)
=
\frac{\E\{u\exp[\ell(a+u)]\}}{\E\{\exp[\ell(a+u)]\}}
=
\frac{N}{D}.
\]
Let $\lambda=\lambda_{\max}(P)$ and split $N$ and $D$ over
$\{\norm{u}\leq r\}$ and $\{\norm{u}>r\}$.

For the tail terms, substitute $u=P^{1/2}v$ with $v\sim\N(0,I_d)$, so that
$\norm{u}\leq\lambda^{1/2}\norm{v}$. Once
$\lambda\leq\min\{1,1/(4c)\}$, the envelope $\exp\{\ell(a+u)\}\leq C\exp(c\norm{u}^2)$ bounds the tail contributions
to $\norm{N}$ and $|D|$ by
\[
C\E\left[
1\{\norm{v}>r\lambda^{-1/2}\}(1+\norm{v})e^{\norm{v}^2/4}
\right]
\leq
C_1e^{-r^2/(8\lambda)},
\]
for all sufficiently small $\lambda$, where $C_1$ depends only on $C$, $d$,
and $r$. This bound is $o(\lambda^k)$ for every $k>0$, so the tail terms are
negligible at any polynomial order.

On $\{\norm{u}\leq r\}$, a third-order Taylor expansion gives
$\ell(a+u)=s^\prime u+\tfrac{1}{2}u^\prime Hu+\rho(u)$ with
$|\rho(u)|\leq L_1\norm{u}^3$, where $L_1$ depends only on $d$
and $L_3$. Since the exponent is bounded on this set, expanding the exponential
yields
\[
\exp[\ell(a+u)]
=
1+s^\prime u
+\tfrac{1}{2} u^\prime(H+ss^\prime)u
+r(u),
\qquad
|r(u)|\leq C_2\norm{u}^3,
\]
where $C_2$ depends on $r$, $\norm{s}$, $\norm{H}$, and $L_1$, but not on $P$.
Take expectations over $\{\norm{u}\leq r\}$. The terms $u$ and
$u\{u^\prime(H+ss^\prime)u\}$ are odd and integrate to zero over
this symmetric set. Moreover,
$\E[uu^\prime1\{\norm{u}\leq r\}]s=Ps$ up to a tail-order
term, and the cubic remainders satisfy
$\E\norm{u}^{3}\leq\{d(d+2)\}^{3/4}\lambda^{3/2}$ and
$\E\norm{u}^{4}\leq d(d+2)\lambda^{2}$. Hence
$N=Ps+O(\norm{P}^2)$ and $D=1+O(\norm{P})$,
and dividing gives
$P\nabla_a\log f(y;a,P)=Ps+O(\norm{P}^2)$,
since $\norm{Ps}=O(\norm{P})$. The posterior-mean expansion follows from Proposition 2. \qed

\subsection{Proof of Proposition 4}
\label{app:implicit-proof}

Coercivity and continuity of $q_t$ give existence of a maximizer, and strict
concavity gives uniqueness. Its first-order condition is
$s_t(a_t+\Delta_t^I)=P_t^{-1}\Delta_t^I$, which is the implicit equation.
Bayes' rule identifies the objective with the log posterior up to a
constant, so the maximizer is the Gaussian-prior posterior mode. For a quadratic log likelihood, completing the square gives a
Gaussian posterior with mean and mode correction
$(P_t^{-1}-H_t)^{-1}s_t(a_t)$.

For the resolvent representation, the fundamental theorem of calculus
gives
\[
 s_t(a_t+d)=s_t(a_t)+\overline H_t(d)d,\qquad
 \overline H_t(d)=\int_0^1\nabla^2\ell_t(a_t+ud)du.
\]
Substitution into the implicit first-order condition gives
$\{P_t^{-1}-\overline H_t(\Delta_t^I)\}\Delta_t^I=s_t(a_t)$, and
inversion proves the representation. For a quadratic likelihood
the average Hessian is constant, and the representation reduces to equation (6). \qed

\subsection{Proof of Theorem 1 (local Bayesian comparison of ESD and ISD)}

Let $X\sim\N_d(0,P)$ with $P=\varepsilon\Omega$. Taylor expansion through cubic order,
Gaussian symmetry, and (\ref{eq:local-envelope-app}) give
\begin{align}
 f(y;a,P)&=L(a)+\frac12\sum_{j,k}P_{jk}L_{jk}(a)
                         +O(\varepsilon^2)L(a),
                 \label{eq:heat-f}\\
 \nabla_a f(y;a,P)&=\nabla L(a)
       +\frac12\sum_{j,k}P_{jk}\nabla L_{jk}(a)
                         +O(\varepsilon^2)L(a),
                 \label{eq:heat-grad}
\end{align}
where $L_{jk}$ denotes the second partial derivative of $L$. The integrated remainder is bounded by
$C L(a)\E[\norm{X}^4\exp(c\norm{X}^2)]=O(\varepsilon^2)L(a)$;
the fifth likelihood derivative supplies this bound for the gradient
expansion. After normalization by the likelihood at the predictive mean, the denominator tends to one, so division is valid for
sufficiently small $\varepsilon$.

Since $L_{jk}/L=H_{jk}+s_js_k$,
\[
 \partial_i(L_{jk}/L)=T_{ijk}+H_{ij}s_k+s_jH_{ik}.
\]
Contracting with $P$ and expanding the ratio of the gradient in
(\ref{eq:heat-grad}) to the density in (\ref{eq:heat-f}) yields
$\nabla_a\log f=s+HPs+\tfrac{1}{2}(T:P)+O(\varepsilon^2)$.
Tweedie's formula under the Gaussian predictive law, equation (2), then gives
\begin{equation}\label{eq:proof-mean-expansion}
 \Delta^B=Ps+PHPs+\frac12P(T:P)+O(\varepsilon^3)
 =\varepsilon\Omega s+\varepsilon^2(u+v)+O(\varepsilon^3),
\end{equation}
with $u=\Omega H\Omega s$ and $v=\tfrac{1}{2}\Omega(T:\Omega)$.

For the implicit solution, define
$F(d,\varepsilon)=d-\varepsilon\Omega s(a+d)$.
Then $F(0,0)=0$ and $\partial_dF(0,0)=I_d$, so
the implicit-function theorem gives a unique local solution
$\Delta_\varepsilon^I=O(\varepsilon)$. Expanding its score gives
\[
 \Delta_\varepsilon^I
 =P\{s+H\Delta_\varepsilon^I+O(\norm{\Delta_\varepsilon^I}^2)\}
 =Ps+PHPs+O(\varepsilon^3).
\]
The negative posterior Hessian at this solution is
$P^{-1}-H(a+\Delta_\varepsilon^I)\succ0$ for small $\varepsilon$,
so it is a local mode. Together with the definition of ESD,
(\ref{eq:proof-mean-expansion}) proves the theorem. Uniform derivative
and eigenvalue bounds give the uniform conclusion. \qed

\subsection{Proof of Corollary 1}

Posterior variance decomposition gives
\[
 \mathcal R_\varepsilon(D)
 =\tr\{\operatorname{var}(\theta| y)\}
        +\norm{D-\Delta_\varepsilon^B}^2.
\]
The posterior variance is common to both rules, so it cancels when
the risks are subtracted. Substituting the errors in
Theorem 1 leaves the leading coefficient
$\norm{u+v}^2-\norm{v}^2=\norm{u}^2+2u^\prime v$.
The product of a leading $O(\varepsilon^2)$ error and its
$O(\varepsilon^3)$ remainder is $O(\varepsilon^5)$, proving the corollary. \qed

\subsection{Proof of Proposition 5 (\ISDp{})}

Write $z=a+\Delta_\varepsilon^I$. Since $z-a=O(\varepsilon)$,
$H(z)=H+O(\varepsilon)$ and $T(z)=T+O(\varepsilon)$. Thus
\[
 W^I=\{P^{-1}-H(z)\}^{-1}
     =P+PHP+O(\varepsilon^3),
 \qquad
 \frac12W^I(T^I:W^I)=\frac12P(T:P)+O(\varepsilon^3).
\]
The inverse-curvature matrix is positive definite for sufficiently small
$\varepsilon$, because $P^{-1}$ dominates the locally bounded Hessian.
Adding the displayed mean--mode correction to the implicit expansion
gives (\ref{eq:proof-mean-expansion}) through second order, which proves
$\Delta_\varepsilon^B-\Delta_\varepsilon^{I\ast}=O(\varepsilon^3)$. For a globally quadratic log likelihood, $T^I=0$
and Proposition 4 gives exactness. \qed

The same correction has a mode-centered Laplace interpretation through
the cubic term in the log posterior. For $X\sim\N(0,W^I)$ its
contribution to the first moment is
\[
 \frac16\E\left[X\sum_{i,j,k}T^I_{ijk}X_iX_jX_k\right]
       =\frac12W^I(T^I:W^I),
\]
by the Gaussian fourth-moment identity. The expansion above establishes
the error order of the resulting correction in the present
small-predictive-covariance regime.

\subsection{Affine transformations}

For $z=C\theta+c$ with nonsingular $C$, the Gaussian predictive
covariance becomes $CPC^\prime$, the score becomes
$(C^{\prime})^{-1}s$, and inverse posterior curvature becomes
$CW^IC^\prime$. The third-derivative contraction satisfies
$W_z^I(T_z^I:W_z^I)=CW^I(T^I:W^I)$.
Hence ESD, ISD, \ISDp{}, and the explicit second-order correction
all transform as $\Delta_z=C\Delta_\theta$. This property concerns
the updates. Euclidean risk in the transformed parameter generally
changes the loss; an equivalent comparison transforms the loss
matrix to $(C^{\prime})^{-1}QC^{-1}$. Nonlinear transformations also
change the posterior-mean target, as the examples in
Section 6 illustrate.

\subsection{Proof of Proposition 6}

Under the Gaussian predictive law and the Fisher-scoring quadratic
surrogate, the surrogate posterior is proportional to $\exp\{\tilde\ell_t(\theta)\}\phi_{P_t}(\theta-a_t)$, whose log density is, up to a term independent of $\theta$,
\[
s_t(a_t)^\prime(\theta-a_t)
-
\frac12(\theta-a_t)^\prime
\{P_t^{-1}+\I_t(a_t)\}
(\theta-a_t).
\]
This is the log kernel of a Gaussian density with precision $W_t^{-1}=P_t^{-1}+\I_t(a_t)$ and mean $a_t+W_ts_t(a_t)$.
By the local precision-discounting condition,
$P_t^{-1}=\delta W_t^{-1}=\delta\{P_t^{-1}+\I_t(a_t)\}$, so
$P_t=\{(1-\delta)/\delta\}\I_t(a_t)^{-1}=\kappa_P\I_t(a_t)^{-1}$ and $W_t=(1-\delta)\I_t(a_t)^{-1}$. Since $\norm{\I_t(a_t)^{-1}}$ is bounded,
$\norm{P_t}=O(\kappa_P)$ and $O(\norm{P_t}^2)=O(\kappa_P^2)$;
combining this with Proposition 3 gives the stated expansion of the posterior mean. \qed

\section{Average Bayesian risk}\label{app:average-risk}

The conditional risk comparison depends on the realized observation.
This appendix establishes an average comparison by controlling the
expansion remainders under the prior-predictive law. Consider the
canonical family
$p(y|\theta)=h(y)\exp\{\theta^\prime t(y)-b(\theta)\}$ on
the full natural parameter space $\R^d$ and let
$\theta\sim\N_d(a,\varepsilon\Omega)$. Full parameter space ensures
that the Gaussian prior assigns mass only to admissible parameters.
For a fixed loss matrix $Q\succ0$, define
$\mathcal B_\varepsilon^j=\E\norm{\theta-a-\Delta^j(Y)}_Q^2$,
where $\norm{x}_Q^2=x^\prime Qx$ and the expectation is under the
joint prior-predictive law. Let $f_\varepsilon$ be its observation
density. All rules use the same $a$ and $\varepsilon\Omega$.

\begin{aassumption}[Integrated regularity]\label{ass:integrated}
The function $b$ is convex and five times continuously differentiable
on $\R^d$. For $k=1,\ldots,5$, its derivatives satisfy
$\norm{\nabla^k b(\theta)}\leq C\exp(c\norm{\theta})$ for finite
$C,c$. For every positive integer $m$,
$\sup_{0<\varepsilon\leq\varepsilon_0}
\E_{f_\varepsilon}\norm{t(Y)}^m<\infty$.
\end{aassumption}

These are sufficient conditions for integrating the expansion. They
hold for finite-support canonical models and for Poisson log intensity.
In the latter case conditional count moments are polynomials in
$e^\theta$, whose Gaussian expectations are finite and uniformly
bounded over the stated covariance interval. Assumption~\ref{ass:integrated} states the integrability conditions referred to in Corollary 2.

\begin{atheorem}[Multivariate average-risk comparison]\label{thm:average-risk}
Under Assumption~\ref{ass:integrated}, let $J=\nabla^2b(a)\succ0$,
$A=\Omega J\Omega$, and
$v=-\tfrac{1}{2}\Omega\{\nabla^3b(a):\Omega\}$.
The Gaussian-prior posterior has a unique mode for every observation,
and
\begin{align*}
 \mathcal B_\varepsilon^E-\mathcal B_\varepsilon^B
 &=\varepsilon^4\{C_Q+v^\prime Qv\}+O(\varepsilon^5),\\
 \mathcal B_\varepsilon^I-\mathcal B_\varepsilon^B
 &=\varepsilon^4v^\prime Qv+O(\varepsilon^5),
                               \\
 \mathcal B_\varepsilon^E-\mathcal B_\varepsilon^I
 &=\varepsilon^4 C_Q+O(\varepsilon^5),\qquad
 C_Q=\tr(QAJA)>0.
\end{align*}
With $d=1$ and $\Omega=Q=1$, this proves Corollary 2.
\end{atheorem}

\begin{proof}
Write $s=t(y)-\nabla b(a)$ and $u=-As$. Convexity gives
$b(a+x)-b(a)\geq\nabla b(a)^\prime x$, so the Gaussian-penalized
log likelihood is strictly concave and tends to $-\infty$. Its mode
exists uniquely. Taking an inner product in the implicit equation and
using monotonicity of $\nabla b$ gives
\begin{equation}\label{eq:mode-bound}
 \norm{\Delta_\varepsilon^I}_{\Omega^{-1}}
       \leq\varepsilon\norm{\Omega^{1/2}s}.
\end{equation}

Let
$\mathcal A_\varepsilon=\{y:\norm{s(y)}\leq\varepsilon^{-1/4}\}$
and $X\sim\N(0,\varepsilon\Omega)$.
For $L_y(\theta)=p(y|\theta)$ on its likelihood support,
convexity and the derivative bounds give, for $k\leq5$,
\[
 \frac{\norm{\nabla^k L_y(a+x)}}{L_y(a)}
 \leq C_k(1+\norm{s}^k)\exp\{s^\prime x+C_k\norm{x}\}.
\]
Gaussian integration of this envelope with $\norm{X}^4$ is bounded by
$C\varepsilon^2(1+\norm{s}^{k+4})$ on $\mathcal A_\varepsilon$;
here $\varepsilon\norm{s}^2\leq\varepsilon^{1/2}$.
Jensen's inequality gives a uniform positive denominator bound:
\[
 \frac{f_\varepsilon(y)}{L_y(a)}
 =\E\exp\{s^\prime X-b(a+X)+b(a)+\nabla b(a)^\prime X\}
 \geq\exp\{-\E[b(a+X)-b(a)]\}\geq c_0>0.
\]
Consequently the Taylor and ratio expansions in
(\ref{eq:heat-f}) and (\ref{eq:heat-grad}) have a remainder satisfying
\[
 \norm{\Delta_\varepsilon^B-\varepsilon\Omega s
                         -\varepsilon^2(u+v)}
 \leq C\varepsilon^3(1+\norm{s}^{m_0})
 \quad\text{on }\mathcal A_\varepsilon
\]
for a fixed integer $m_0$. Equation~(\ref{eq:mode-bound}) makes the
implicit displacement uniformly small on this set, so its expansion
has a bound of the same type. The moment assumption bounds both
remainders in $L^2(f_\varepsilon)$ on this set by $C\varepsilon^3$.

It remains to bound the contribution of $\mathcal A_\varepsilon^c$.
For any sufficiently large integer $M$, Markov's inequality and the
score moment assumption give
$\Pr_{f_\varepsilon}(\mathcal A_\varepsilon^c)
=O(\varepsilon^{M/4})$.
Conditional Jensen gives
$\E_{f_\varepsilon}\norm{\Delta_\varepsilon^B}^4
\leq\E\norm{\theta-a}^4=O(\varepsilon^2)$.
Cauchy--Schwarz then bounds the squared $L^2$ contribution of
$\Delta_\varepsilon^B1_{\mathcal A_\varepsilon^c}$ by
$O(\varepsilon^{1+M/8})$. Taking $M\geq40$ gives
$O(\varepsilon^6)$. Equation~(\ref{eq:mode-bound}) and higher score
moments give the same or smaller order for the implicit displacement
and the polynomial expansion terms. Hence
\begin{align*}
 \Delta_\varepsilon^B
 &=\varepsilon\Omega s+\varepsilon^2(u+v)
               +O_{L^2(f_\varepsilon)}(\varepsilon^3),\\
 \Delta_\varepsilon^I
 &=\varepsilon\Omega s+\varepsilon^2u
               +O_{L^2(f_\varepsilon)}(\varepsilon^3).
\end{align*}

The canonical conditional moment
identities and Gaussian expansion of $b$ imply
\[
 \E_{f_\varepsilon}s=O(\varepsilon),\qquad
 \E_{f_\varepsilon}(ss^\prime)=J+O(\varepsilon).
\]
Thus $\E(u^\prime Qv)=O(\varepsilon)$ and
$\E(u^\prime Qu)=\tr(QAJA)+O(\varepsilon)$.
Weighted posterior variance decomposition gives
$\mathcal B_\varepsilon^j-\mathcal B_\varepsilon^B
=\E_{f_\varepsilon}\norm{\Delta^j-\Delta^B}_Q^2$.
Substituting the integrated expansions proves the stated formulas.
Strict positivity follows from
$C_Q=\norm{Q^{1/2}AJ^{1/2}}_F^2>0$.
\end{proof}

Section 6.4 illustrates the distinction between
conditional and average risk. Its Bernoulli log-odds example has
$\Omega=Q=1$ and $J=0.16$, so the positive average coefficient is
$C_Q=J^3=0.004096$ despite a conditional reversal after success.

\section{Comparisons at finite predictive covariance}
\label{app:finite-risk}

At a fixed predictive covariance, a bound on the remaining mean--mode
displacement gives a verifiable sufficient condition for the posterior
risk comparison. Write
$\widehat\theta_E=a+Ps(a)$, $\widehat\theta_I=a+\Delta^I$, and
$m=\E(\theta| y)$. For Euclidean squared loss let
$\mathcal R_y(z)=\E(\norm{\theta-z}^2| y)$.

\begin{atheorem}[Finite-covariance comparison certificate]
\label{thm:finite-risk}
Let $q(\theta)=\ell(\theta)-\tfrac{1}{2}(\theta-a)^\prime P^{-1}(\theta-a)$
be three times continuously differentiable, with
$-\nabla^2q(\theta)\succeq\alpha I_d$ for all $\theta$, where $\alpha>0$.
Suppose $H=\nabla^2\ell$ satisfies
$\norm{H(x)-H(z)}_{\op}\leq M\norm{x-z}$ globally. Set
$W=(P^{-1}-H(\widehat\theta_I))^{-1}$ and
\begin{equation*}
 B=\frac{Md\norm{W}_{\op}}{2\alpha}.
\end{equation*}
Then $\norm{m-\widehat\theta_I}\leq B$. If
$d_{IE}=\widehat\theta_I-\widehat\theta_E$, then exactly
\begin{equation}\label{eq:finite-risk-identity}
 \mathcal R_y(\widehat\theta_E)-\mathcal R_y(\widehat\theta_I)
 =\norm{d_{IE}}^2+2d_{IE}^\prime(m-\widehat\theta_I),
\end{equation}
and hence
\begin{equation}\label{eq:finite-risk-certificate}
 \mathcal R_y(\widehat\theta_E)-\mathcal R_y(\widehat\theta_I)
 \geq\norm{d_{IE}}\{\norm{d_{IE}}-2B\}.
\end{equation}
In particular, $\norm{d_{IE}}>2B$ certifies strictly smaller posterior risk
for ISD without evaluating the posterior mean.
\end{atheorem}

The bound combines the curvature at the mode, through $W$, with global bounds
on the posterior dispersion and on the third derivative of the log likelihood.
It vanishes for quadratic log likelihoods. It also applies to nonconcave
likelihoods, including Student-$t$ location models, provided the Gaussian prior
makes the posterior strongly log concave. The certificate is sufficient and may
be conservative.

For Bernoulli log odds,
$\sup_\theta|\ell'''(\theta)|=1/(6\sqrt3)$ and $\alpha=1/P$, giving
$B=PW/(12\sqrt3)$. Poisson log intensity has unbounded third derivative,
so Theorem~\ref{thm:finite-risk} does not apply with a finite global $M$.
Its structure instead gives a sharper one-sided enclosure, which is the result summarized in Section 5.6 and used in the Poisson example of Section 6.3.

\begin{aproposition}[Poisson posterior-mean enclosure]
\label{prop:poisson-enclosure}
Suppose $Y|\theta\sim\operatorname{Poisson}(e^\theta)$ and
$\theta\sim\N(a,P)$, with $P>0$. Let $z_0$ and $z_-$ be the unique roots
\begin{equation}\label{eq:poisson-bracket}
 z_0+Pe^{z_0}=a+Py,\qquad
 z_-+Pe^{z_-+P/2}=a+Py.
\end{equation}
Then $z_0=\widehat\theta_I$ and
\begin{equation*}
 z_-\leq m\leq z_0,\qquad
 0\leq z_0-m\leq z_0-z_-<P/2.
\end{equation*}
For fixed $a,y$, both $z_0-m$ and $z_0-z_-$ equal
$\tfrac{1}{2}e^aP^2+O(P^3)$ as $P\downarrow0$.
If $y>e^a$, ISD has strictly smaller posterior squared-error risk than ESD
for every $P>0$. If $y=e^a$, the two updates coincide.
\end{aproposition}

The enclosure uses the sub-Gaussian concentration implied by the Gaussian
prior and the concave Poisson log likelihood. It is valid at every predictive
variance and is sharp to the leading mean--mode order. Two monotone scalar
root solves suffice, so no posterior integration or estimated parameter is
required. More generally, for any scalar enclosure $l\leq m\leq u$, the
endpoints of the risk-difference interval are
\begin{equation}\label{eq:enclosure-risk-interval}
 d_{IE}(2l-\widehat\theta_I-\widehat\theta_E),\qquad
 d_{IE}(2u-\widehat\theta_I-\widehat\theta_E),
\end{equation}
ordered from smaller to larger. A positive lower endpoint certifies ISD;
a negative upper endpoint certifies ESD. An interval containing zero leaves
the comparison unresolved. These bounds are diagnostics for the accuracy of
ESD and ISD; they do not define new updates.

\subsection{Proof of Theorem~\ref{thm:finite-risk}}

Let
$z=\widehat\theta_I$, $X=\theta-z$, and $J_z=W^{-1}$.
Integration by parts under the posterior gives $\E\nabla q(\theta)=0$ and
$\E[X^\prime\nabla q(z+X)]=-d$. Strong concavity implies
$X^\prime\nabla q(z+X)\leq-\alpha\norm{X}^2$, so
$\E\norm{X}^2\leq d/\alpha$.
Strong concavity supplies a Gaussian tail bound. Together with the
globally Lipschitz Hessian, it justifies the integration-by-parts steps.
To relate this dispersion bound to the mean--mode displacement, apply
Taylor's formula to the score:
\[
 s(z+X)=s(z)+H(z)X+r(X),\qquad
 \norm{r(X)}\leq\tfrac{1}{2}M\norm{X}^2.
\]
Subtracting the mode equation from $\E\nabla q(\theta)=0$ yields the exact
identity $J_z\E X=\E r(X)$. Thus
$\norm{m-z}\leq\norm{W}_{\op}M\E\norm{X}^2/2\leq B$.
Posterior variance decomposition proves
(\ref{eq:finite-risk-identity}); Cauchy--Schwarz gives
(\ref{eq:finite-risk-certificate}). \qed

\subsection{Proof of Proposition~\ref{prop:poisson-enclosure}}

The log posterior density has second derivative $-P^{-1}-e^\theta$. Every exponential
tilt of this density is therefore also strongly log concave with curvature
at least $P^{-1}$. The scalar covariance bound gives
variance at most $P$ under each such tilt. To verify the variance bound,
let $q_t$ be a tilted log density, and let $X,X'$ be independent draws from
it. Then
\[
 -1=\E[(X-\E X)q_t'(X)]
 =\tfrac{1}{2}\E[(X-X')\{q_t'(X)-q_t'(X')\}]
 \leq-P^{-1}\V(X).
\]
Writing $K(t)=\log\E(e^{t\theta}| y)$, we have $K''(t)\leq P$.
Integrating twice from zero gives
$\E(e^{\theta-m}| y)\leq e^{P/2}$. Jensen's inequality supplies the
lower bound, so $e^m\leq\E(e^\theta| y)\leq e^{m+P/2}$.
Combined with the posterior score identity
$m+P\E(e^\theta| y)=a+Py$ and the monotonicity of the two functions in
(\ref{eq:poisson-bracket}), these bounds give $z_-\leq m\leq z_0$.
Evaluating the second function at $z_0-P/2$ gives $a+Py-P/2$, proving
$z_->z_0-P/2$.

At fixed $a,y$, expansion of the two root equations gives
$z_0-z_-=(e^a/2)P^2+O(P^3)$; Theorem 1 gives
the same leading coefficient for $z_0-m$. If $y>e^a$, concavity implies
$a<z_0<\widehat\theta_E=a+P(y-e^a)$, while $m\leq z_0$.
Both squared distances from $m$ are therefore ordered strictly in favor of
ISD. For $y=e^a$, both updates equal $a$. \qed

\section{Update sensitivity and recursive error}
\label{app:dynamic-bounds}

This appendix shows how errors in the posterior-mean correction accumulate
when updating and prediction are repeated. The accumulation depends on the
persistence of the prediction step and on the strength of the observation
update.

\subsection{Sensitivity of Bayesian and implicit updating}

Write $F^B(a)=m(a;P)$ and $F^I(a)=\widehat\theta_I(a;P)$, holding $P$ fixed
when changing $a$. Let $V(a;P)$ be the posterior covariance and
$W^I(a;P)=\{P^{-1}-H(F^I(a))\}^{-1}$.

\begin{aproposition}[Posterior uncertainty and update sensitivity]
\label{prop:update-sensitivity}
Under the Gaussian predictive law and the differentiation and moment
conditions of the covariance identity in Appendix~\ref{app:gaussian-form},
and with a unique differentiable implicit mode of positive curvature,
\begin{equation}\label{eq:update-sensitivity}
 \frac{\partial F^B(a)}{\partial a^\prime}=V(a;P)P^{-1},\qquad
 \frac{\partial F^I(a)}{\partial a^\prime}=W^I(a;P)P^{-1}.
\end{equation}
Suppose moreover that $-H(\theta)\succeq J_0\succeq0$ for every $\theta$.
Both maps are globally Lipschitz in the norm
$\norm{x}_{P^{-1}}=(x^\prime P^{-1}x)^{1/2}$, with coefficient at most
\begin{equation*}
 \lambda(P,J_0)
 =\{1+\lambda_{\min}(P^{1/2}J_0P^{1/2})\}^{-1}.
\end{equation*}
\end{aproposition}

In the scalar case, (\ref{eq:update-sensitivity}) reduces to the sensitivities $V/P$ and $W^I/P$ displayed in Section 5.8. Thus, the sensitivity of each update to the predictive mean is determined
by posterior uncertainty, and it is bounded under the curvature condition. The covariance bound used in the proof is the classical
Brascamp--Lieb inequality; see \citet{SaumardWellner:2014}. ISD's proximal
contraction is established by \citet{LangeVanOsVanDijk:2026}. We use these
sensitivities to bound the approximation error relative to the exact mean
recursion. If $P$ depends on $a$, additional derivatives appear, and
(\ref{eq:update-sensitivity}) does not apply.

\subsection{Mean-square error under repeated updating}

For a prescribed common Gaussian predictive family, let
$F^B_{t,\varepsilon}(a)$ be its exact posterior-mean update and
$F^j_{t,\varepsilon}(a)=a+D^j_{t,\varepsilon}(a)$ an approximation. Consider
\begin{align*}
 a_{t+1}^B&=g_t\{F^B_{t,\varepsilon}(a_t^B)\},
 \\*
 a_{t+1}^j&=g_t\{F^j_{t,\varepsilon}(a_t^j)\}.
\end{align*}
The initial means and $a_t^B,a_t^j$ are $\F_{t-1}$-measurable. These
Gaussian-reset mean recursions define the comparison target. They need not coincide with the exact
filter of a latent state-space model. All norms below are taken in one
fixed Euclidean or weighted Euclidean metric.

\begin{atheorem}[Mean-square propagation of Bayesian approximation error]
\label{thm:dynamic-error}
Suppose all iterates remain in a common set on which there are nonnegative
$\F_t$-measurable bounds $q_{t,\varepsilon}$ and
$\eta_{t,\varepsilon}^j$ satisfying
\begin{align*}
 \norm{g_t(F^B_{t,\varepsilon}(a))-g_t(F^B_{t,\varepsilon}(b))}
 &\leq q_{t,\varepsilon}\norm{a-b},\\
 \norm{g_t(F^B_{t,\varepsilon}(a))-g_t(F^j_{t,\varepsilon}(a))}
 &\leq\eta_{t,\varepsilon}^j.
\end{align*}
Assume, uniformly in $t$,
\begin{equation}\label{eq:dynamic-conditional-moments}
 \E(q_{t,\varepsilon}^2|\F_{t-1})\leq\bar q_\varepsilon^2<1,
 \qquad
 \E\{(\eta_{t,\varepsilon}^j)^2|\F_{t-1}\}
 \leq K_j^2\varepsilon^{2(r_j+1)},
\end{equation}
where $K_j$ is independent of $\varepsilon$. For
$e_t^j=\{\E\norm{a_t^B-a_t^j}^2\}^{1/2}$,
\begin{equation}\label{eq:dynamic-bound}
 e_t^j\leq\bar q_\varepsilon^{t-1}e_1^j+
 \frac{K_j\varepsilon^{r_j+1}(1-\bar q_\varepsilon^{t-1})}
 {1-\bar q_\varepsilon}.
\end{equation}
No independence between the contraction and local-error bounds is required.
If the bounds hold deterministically, the same inequality holds pathwise.
\end{atheorem}

The moment conditions must be verified separately, because they do not
follow from the pointwise expansion. The proof requires contraction of the
Bayesian map, not of \ISDp{}.

With equal initialization, $e_1^j=0$, and (\ref{eq:dynamic-bound}) gives
\begin{equation}\label{eq:recursive-method-orders}
 \sup_{t\geq1}e_t^j\leq
 \frac{K_j\varepsilon^{r_j+1}}{1-\bar q_\varepsilon},
\end{equation}
which is Corollary 3. For ESD and ISD take $r_E=r_I=1$; for the explicit second-order and \ISDp{}
updates take $r_2=r_{I\ast}=2$, provided their local errors satisfy
(\ref{eq:dynamic-conditional-moments}). If the contraction gap is bounded
away from zero, the recursive bounds have the same orders as the one-step
errors. If instead $1-\bar q_\varepsilon\asymp\varepsilon^\kappa$, with
$\kappa>0$, their orders are $O(\varepsilon^{r_j+1-\kappa})$.

For example, let $P=\varepsilon\Omega$ with fixed $\Omega$, use the
$\Omega^{-1}$ norm, and suppose $g_t$ has Lipschitz coefficient at most
$\rho_\varepsilon$. If
$-H_t(\theta)\succeq j_0\Omega^{-1}$ globally, with $j_0>0$, then
Proposition~\ref{prop:update-sensitivity} permits
\begin{equation*}
 \bar q_\varepsilon=\frac{\rho_\varepsilon}{1+\varepsilon j_0}.
\end{equation*}
Thus informative observations can supply contraction even when
$g_t$ is the identity. In that case the gap is of order $\varepsilon$, and
the resulting uniform bounds are $O(\varepsilon)$ for ESD and ISD and
$O(\varepsilon^2)$ for \ISDp{}, rather than their one-step orders.
The same conclusion holds for $\rho_\varepsilon=1-c\varepsilon$ with
$c\geq0$. A fixed $\rho<1$ instead preserves the one-step orders.

These orders are upper bounds, and a particular score-driven model can do
better. The factor $(1-\bar q_\varepsilon)^{-1}$ cannot in general be
removed, because the scalar recursion
$e_{t+1}=\bar q_\varepsilon e_t+K\varepsilon^{r+1}$ with $e_1=0$
attains it in the limit. In persistent time series, a small one-step error can therefore become a
larger error after many updates.
Our target is the prescribed Bayesian mean recursion, whereas the tracking
bounds in \citet{DonkervanHeelLangevanOsvanDijk:2026} concern the
(pseudo-)true parameter path.

\subsection{Proof of Proposition~\ref{prop:update-sensitivity}}

Differentiate $m=a+P\nabla_a\log f$ and use the covariance identity of Appendix~\ref{app:gaussian-form} to obtain $D_a m=VP^{-1}$.
Implicit differentiation of $F^I-a=Ps(F^I)$ gives
$D_aF^I=(P^{-1}-H(F^I))^{-1}P^{-1}$.
The curvature condition and the Brascamp--Lieb covariance inequality
imply $V\preceq(P^{-1}+J_0)^{-1}$. The curvature condition gives the
same bound directly for $W^I$.
In coordinates $z=P^{-1/2}a$, the two Jacobians are the symmetric
positive semidefinite matrices $P^{-1/2}VP^{-1/2}$ and
$P^{-1/2}W^IP^{-1/2}$. Their largest eigenvalues are at most
$\lambda(P,J_0)$, giving the Lipschitz bound by integration along line
segments. \qed

\subsection{Proof of Theorem~\ref{thm:dynamic-error} and Corollary 3}

Add and subtract $g_t(F^B_{t,\varepsilon}(a_t^j))$. The pathwise error
$E_t=\norm{a_t^B-a_t^j}$ satisfies
$E_{t+1}\leq q_{t,\varepsilon}E_t+\eta_{t,\varepsilon}^j$.
Since $E_t$ is $\F_{t-1}$-measurable,
\[
 \E(q_{t,\varepsilon}^2E_t^2)
 =\E\{E_t^2\E(q_{t,\varepsilon}^2|\F_{t-1})\}
 \leq\bar q_\varepsilon^2\E E_t^2.
\]
Minkowski's inequality gives
$e_{t+1}^j\leq\bar q_\varepsilon e_t^j+K_j\varepsilon^{r_j+1}$.
Iterating this scalar inequality proves (\ref{eq:dynamic-bound}).
Under equal initialization, taking the supremum gives
(\ref{eq:recursive-method-orders}). Substituting the contraction-gap order
proves the persistence statement. The same scalar recursion with equality
shows sharpness of the accumulation factor for this abstract class of maps,
without asserting sharpness for every score-driven model. \qed

\section{Numerical implementation and additional results}\label{app:numerics}

The numerical results and diagnostics in Section 7 and in this appendix are computed with the Julia package TweedieScore.jl \citep{HansenTong:2026Julia}, available at \url{https://github.com/reinhardhansen/TweedieScore}. Its replication folder reruns the experiments and rebuilds every table, Figure 1, and every number that the text quotes from the experiments. Figure 1 uses the same static simulation run as Table 3.

\subsection{Static design and posterior integration}\label{app:static-design}

Each replication in Section 7.1 draws $\theta\sim\N(a,P)$, followed by $Y$ from its
conditional observation model.
For Poisson log intensity,
$Y|\theta\sim\operatorname{Poisson}(e^\theta)$ and $a=\log4$.
For Gaussian log variance, $Y|\theta\sim\N(0,e^\theta)$ and $a=0$.
For Student-$t_5$ location, $Y=\theta+\sqrt{3/5}E$ with $E\sim t_5$
and $a=0$. Each model uses 100,000 replications at every covariance,
with seed 20260814 and common draws across methods. The covariance grid is
\[
 \begin{aligned}
 \mathcal P=\{&0.00015625,0.0003125,0.000625,0.00125,0.0025,0.005,\\
              &0.01,0.02,0.04,0.08,0.16,0.32\}.
 \end{aligned}
\]

Using physicists' Gauss--Hermite nodes $x_j$ and weights $w_j$, we
evaluate posterior expectations as
\begin{equation*}
 \E\{g(\theta)| y\}\simeq
 \frac{\sum_{j=1}^J w_j g(a+\sqrt{2P}x_j)
                            p(y| a+\sqrt{2P}x_j)}
      {\sum_{j=1}^J w_j p(y| a+\sqrt{2P}x_j)}.
\end{equation*}
For numerical stability, the implementation forms the weights on the
log scale and subtracts their largest value before exponentiation.
The main run uses $J=80$. Repeating the calculation with $J=120$
changes any reported RMSE by at most
$1.47\times10^{-7}$ and any estimated slope by at most
$7.40\times10^{-8}$.

We report approximation RMSE relative to the posterior mean,
$[N^{-1}\sum_{i=1}^N(\Delta_i^j-\Delta_i^B)^2]^{1/2}$.
The corresponding population mean squared error equals the excess
Bayesian risk relative to the posterior mean, by posterior variance decomposition.
The RMSE therefore measures approximation to the posterior mean, not
estimation error for the latent state. The small-$P$ slopes come
from regressing log RMSE on log $P$ over the five values at most $0.0025$.
The slope of the explicit second-order update for log variance is the
least precise estimate, because its squared errors are heavy-tailed. The
Monte Carlo standard error of each RMSE is $6\%$ to $17\%$, compared with
$0.5\%$ for \ISDp{}, so the estimated slope has a Monte Carlo standard
error of about $0.05$.

The pointwise expansions imply generic slopes two for ESD and ISD and
three for \ISDp{} and explicit second order, provided the leading
coefficient is nonzero and a square-integrable remainder permits
averaging over observations.
Theorem~\ref{thm:average-risk} supplies an integrated ESD--ISD result
for its canonical family. In the other numerical designs, we use the
estimated slopes as empirical diagnostics of the local approximation
orders.

\subsection{Derivatives and the implicit solution}

For Poisson log intensity, the three derivatives of the log likelihood
are $s=y-e^\theta$, $H=-e^\theta$, and $T=-e^\theta$.
For Gaussian log variance they are
$s=(y^2e^{-\theta}-1)/2$, $H=-y^2e^{-\theta}/2$, and
$T=y^2e^{-\theta}/2$.
In both models, the Gaussian-prior penalized objective is globally
strictly concave.

For Student-$t_\nu$ location with scale $\sigma$, put $r=y-\theta$ and
$d_0=\nu\sigma^2+r^2$. Then
\begin{align}
 s(\theta)&=\frac{(\nu+1)r}{d_0},\nonumber\\
 H(\theta)&=\frac{(\nu+1)(r^2-\nu\sigma^2)}{d_0^2},\nonumber\\
 T(\theta)&=\frac{2(\nu+1)r(r^2-3\nu\sigma^2)}{d_0^3}.
                                                 \label{eq:t-derivatives}
\end{align}
Setting $\sigma^2=(\nu-2)/\nu$ gives unit-variance observation errors.
Since $\sup_\theta H(\theta)=(\nu+1)/\{8(\nu-2)\}$,
$P<8(\nu-2)/(\nu+1)$ is sufficient for strict concavity of the
Gaussian-prior log posterior. The main $t_5$ grid satisfies $P\leq0.32<4$.

We solve the implicit equation by safeguarded Newton iterations,
using bisection when needed. The largest residual in the main run is
$4.91\times10^{-13}$. The smallest realized curvature denominator
$1-PH(a+\Delta^I)$ is $0.920$, ensuring that the mode correction is
well defined throughout the reported grid. In one dimension,
$W^I=P/\{1-PH(a+\Delta^I)\}$ and
$\Delta^{I\ast}=\Delta^I+(W^I)^2T(a+\Delta^I)/2$.

\subsection{Tail-thickness comparison}

To examine sensitivity to tail thickness, we repeat the Student-$t$
experiment for $\nu=3$ and $10$, retaining unit-variance observation
errors. The sufficient
strict-concavity bounds are $2$ and $64/11$, both above the largest
reported covariance. Table~\ref{tab:student-robustness} shows that the
small-variance slopes remain close to two and three for the respective
pairs of methods. Accuracy at finite variance nevertheless varies with
tail thickness, consistent with the local scope of the expansion.

\begin{table}[!htbp]
\tableformat
\caption{Static Student-$t$ robustness across tail thickness}
\label{tab:student-robustness}
\begin{papertabular}{crrrr}
\toprule
$\nu$ & ESD & ISD & \ISDp{} & Explicit 2nd \\
\midrule
\multicolumn{5}{@{}l@{}}{\emph{Panel A: Small-$P$ convergence slopes}} \\
3 & 1.996 & 1.994 & 2.986 & 2.993 \\
5 & 1.999 & 1.998 & 2.995 & 2.998 \\
10 & 2.000 & 1.999 & 2.999 & 3.000 \\
\midrule
\multicolumn{5}{@{}l@{}}{\emph{Panel B: RMSE at $P=0.16$}} \\
3 & $6.726\times10^{-2}$ & $2.357\times10^{-2}$ & $4.589\times10^{-3}$ & $4.455\times10^{-2}$ \\
5 & $3.320\times10^{-2}$ & $9.826\times10^{-3}$ & $9.486\times10^{-4}$ & $1.044\times10^{-2}$ \\
10 & $2.403\times10^{-2}$ & $4.559\times10^{-3}$ & $2.080\times10^{-4}$ & $4.578\times10^{-3}$ \\
\bottomrule
\end{papertabular}
\begin{floatnotes}
Notes: The Student-$t_\nu$ observation error is standardized to unit variance.
RMSE is relative to
the exact Gaussian-prior posterior mean. Small-$P$ slopes use the five
values $P\leq0.0025$. Each entry uses 100,000 prior-predictive draws and
80-node Gauss--Hermite quadrature.
\end{floatnotes}
\end{table}

\FloatBarrier
\subsection{Deterministic risk experiments}\label{app:risk-numerics}

The scalar models in Section 7.2 use Bernoulli log odds and Poisson log intensity
with $a=\log4$ and $P=\varepsilon$. In the bivariate model,
the two observations are conditionally independent with success
probabilities $\operatorname{logit}^{-1}(x_i^\prime\theta)$, where
$x_i^\prime$ is row $i$ of $X$. The design is
\begin{equation}\label{eq:bivariate-design}
 a=\begin{pmatrix}\log4\\-0.4\end{pmatrix},\qquad
 X=\begin{pmatrix}1&0.5\\-0.4&1\end{pmatrix},\qquad
 P=\varepsilon\begin{pmatrix}1&0.6\\0.6&1.4\end{pmatrix}.
\end{equation}
The off-diagonal entries in both the covariance and the likelihood
information make the comparison depend on the interaction between
parameters, so the design tests the multivariate coefficient in Theorem~\ref{thm:average-risk} with ordinary Euclidean squared loss. We use
\[
\varepsilon\in\{0.0001,0.0003,0.001,0.003,0.01,0.03,
0.10,0.16,0.32,0.50\}.
\]

We enumerate all Bernoulli outcomes and all four bivariate observation
pairs, using 80 Gaussian nodes per dimension for integration. For
Poisson, quadrature is centered at the implicit mode and scaled by
$\sqrt{W^I}$; retaining the exact posterior reweighting factor preserves
the target integral. The count sum includes zero through 4096.
We check the calculation using 120 nodes per dimension and counts
through 8192. These deterministic calculations have no Monte Carlo
sampling error. We compute risk differences from the errors relative to the posterior mean, which avoids subtracting nearly equal total risks.

As in the static experiment, the reported RMSE is the square root of
excess Bayesian risk. The posterior variance decomposition gives
\[
 \mathcal B_\varepsilon^j
 =\E_f\operatorname{tr}\{\operatorname{var}(\theta| Y)\}
  +\E_f\|\Delta^j(Y)-\Delta^B(Y)\|^2.
\]
The Poisson model also permits an independent closed-form check:
\begin{equation*}
 \mathcal B_P^E
 =P+P^2\{e^{2a+2P}-(2e^a+1)e^{a+P/2}+e^{2a}\}.
\end{equation*}
The expression follows from
$\operatorname{cov}(\theta,Y)=P\E(e^\theta)$ and the first
two moments of the Poisson--lognormal predictive distribution.
The numerical audit checks this identity together with the excess-risk
decomposition, the asymptotic coefficients, and each posterior-mean
bound. Refining the quadrature changes every reported RMSE by less than
$10^{-14}$; the closed-form Poisson risk agrees within the same
tolerance. Central-difference checks verify the scalar update
sensitivities to within $3\times10^{-11}$.

Table~\ref{tab:risk-coefficient} shows the normalized average risk improvement approaching the positive coefficient predicted by Theorem~\ref{thm:average-risk}. At the smallest displayed scale $\varepsilon=10^{-3}$, the relative discrepancy is below $0.07\%$ in every model. Table~\ref{tab:risk-rmse} extends the comparison to larger scales and reports all four methods. In the bivariate design at $\varepsilon=0.32$, ESD and ISD have approximation RMSEs $0.03913$ and $0.008698$; \ISDp{} further reduces the RMSE to $0.001347$, compared with $0.006800$ for explicit second order.

\begin{table}[!htbp]
\tableformat
\caption{Average risk improvement, normalized by $\varepsilon^4$}
\label{tab:risk-coefficient}
\begin{papertabular}{lrrrr}
\toprule
Model & $\varepsilon=0.001$ & $0.01$ & $0.10$ & Limit \\
\midrule
Bernoulli & 0.004099 & 0.004121 & 0.004315 & 0.004096 \\
Poisson & 63.967859 & 63.668216 & 60.643563 & 64.000000 \\
Bivariate logistic & 0.162499 & 0.162125 & 0.156890 & 0.162538 \\
\bottomrule
\end{papertabular}
\begin{floatnotes}
Notes: Entries are $(\mathcal B_\varepsilon^E-\mathcal B_\varepsilon^I)/
\varepsilon^4$. The limit is $C_Q=\operatorname{tr}(QAJA)$ in
Theorem~\ref{thm:average-risk}. The scalar designs have $a=\log4$,
$P=\varepsilon$; the bivariate design is given in
(\ref{eq:bivariate-design}). All expectations use the common Gaussian
prior-predictive law.
\end{floatnotes}
\end{table}

\begin{table}[!htbp]
\tableformat
\caption{Posterior-mean RMSE in the deterministic risk experiments}
\label{tab:risk-rmse}
\begin{papertabular}{lrrrrr}
\toprule
Model & $\varepsilon$ & ESD & ISD & \ISDp{} & Explicit 2nd \\
\midrule
Bernoulli & 0.01 & $8.000\times10^{-6}$ & $4.774\times10^{-6}$ & $1.098\times10^{-8}$ & $1.884\times10^{-8}$ \\
Bernoulli & 0.1 & $7.987\times10^{-4}$ & $4.543\times10^{-4}$ & $1.042\times10^{-5}$ & $1.842\times10^{-5}$ \\
Bernoulli & 0.32 & $8.071\times10^{-3}$ & $4.121\times10^{-3}$ & $2.949\times10^{-4}$ & $5.832\times10^{-4}$ \\
Poisson & 0.01 & $8.192\times10^{-4}$ & $1.854\times10^{-4}$ & $3.568\times10^{-7}$ & $2.722\times10^{-5}$ \\
Poisson & 0.1 & $7.854\times10^{-2}$ & $1.021\times10^{-2}$ & $2.630\times10^{-5}$ & $2.516\times10^{-2}$ \\
Poisson & 0.32 & $7.750\times10^{-1}$ & $3.738\times10^{-2}$ & $5.195\times10^{-4}$ & $7.911\times10^{-1}$ \\
Bivariate logistic & 0.01 & $4.222\times10^{-5}$ & $1.270\times10^{-5}$ & $7.774\times10^{-8}$ & $2.107\times10^{-7}$ \\
Bivariate logistic & 0.1 & $4.116\times10^{-3}$ & $1.120\times10^{-3}$ & $6.508\times10^{-5}$ & $2.121\times10^{-4}$ \\
Bivariate logistic & 0.32 & $3.913\times10^{-2}$ & $8.698\times10^{-3}$ & $1.347\times10^{-3}$ & $6.800\times10^{-3}$ \\
\bottomrule
\end{papertabular}
\begin{floatnotes}
Notes: Entries are $\{\E_f\|\Delta^j-\Delta^B\|^2\}^{1/2}$,
where $f$ is the observation distribution induced by the Gaussian prior.
The scalar designs use $a=\log4$, $P=\varepsilon$; the bivariate
design uses Euclidean loss and (\ref{eq:bivariate-design}).
The Poisson values use deterministic integration and differ slightly
from the Monte Carlo estimates in Table 3.
\end{floatnotes}
\end{table}

The finite-covariance comparison uses the symmetric Bernoulli
interval $[\widehat\theta_I-B,\widehat\theta_I+B]$ from
Theorem~\ref{thm:finite-risk} and the Poisson interval $[z_-,z_0]$
from Proposition~\ref{prop:poisson-enclosure}.
We declare a comparison only when the entire risk-difference interval
in (\ref{eq:enclosure-risk-interval}) has the same sign.
Table~\ref{tab:risk-certificate} reports the probability of
establishing lower ISD risk, the average interval width, and
the actual mean--mode gap. The Poisson enclosure establishes lower ISD risk for observations carrying about $81\%$ to $89\%$ of the predictive probability across the reported variances. The symmetric Bernoulli bound is more conservative. It resolves the outcome favoring ISD but leaves the opposing outcome unresolved. Independent posterior integration agrees
with every declared comparison.

\begin{table}[!htbp]
\tableformat
\caption{Finite-covariance bounds establishing lower ISD risk}
\label{tab:risk-certificate}
\begin{papertabular}{lrrrr}
\toprule
Model & $P$ & Certified probability & Interval width & Mean--mode gap \\
\midrule
Bernoulli & 0.01 & 0.2005 & $9.607\times10^{-6}$ & $4.774\times10^{-6}$ \\
Bernoulli & 0.1 & 0.2047 & $9.471\times10^{-4}$ & $4.543\times10^{-4}$ \\
Bernoulli & 0.32 & 0.2143 & $9.373\times10^{-3}$ & $4.121\times10^{-3}$ \\
Bernoulli & 0.5 & 0.2215 & $2.227\times10^{-2}$ & $9.121\times10^{-3}$ \\
Poisson & 0.01 & 0.8084 & $1.928\times10^{-4}$ & $1.853\times10^{-4}$ \\
Poisson & 0.1 & 0.8351 & $1.463\times10^{-2}$ & $1.019\times10^{-2}$ \\
Poisson & 0.32 & 0.8707 & $9.202\times10^{-2}$ & $3.725\times10^{-2}$ \\
Poisson & 0.5 & 0.8870 & $1.687\times10^{-1}$ & $5.162\times10^{-2}$ \\
\bottomrule
\end{papertabular}
\begin{floatnotes}
Notes: Certified probability is the predictive probability that the
lower bound on ESD-minus-ISD posterior risk is positive.
The last two columns report $\E_f(u-l)$ and
$\E_f|m-\widehat\theta_I|$. Bernoulli uses
$l=\widehat\theta_I-B$, $u=\widehat\theta_I+B$; Poisson uses
$l=z_-$, $u=z_0$. An unresolved comparison does not imply equal risk.
\end{floatnotes}
\end{table}
\FloatBarrier

\subsection{Design of the recursive experiment}\label{app:recursive-design}

The static experiments reset the predictive distribution in every
replication, so approximation errors do not affect later updates. We study
this feedback in a dynamic location model with Student-$t$ observations,
whose score is bounded and whose log-likelihood curvature changes sign; see
\citet{HarveyLuati:2014} for this model. The theory is not restricted to
heavy-tailed location models.

The data-generating process is
\begin{equation*}
 \theta_t=\varphi\theta_{t-1}+\sqrt{q}\eta_t,
 \qquad
 Y_t=\theta_t+\sigma_\nu E_t,
 \qquad
 \eta_t\sim\N(0,1),\quad E_t\sim t_\nu,
\end{equation*}
with mutually independent innovations. We set
$\sigma_\nu=\{(\nu-2)/\nu\}^{1/2}$, so the observation disturbance has unit
variance for every $\nu>2$. If $v_\theta$ denotes the stationary state
variance, then $q=v_\theta(1-\varphi^2)$ and the signal-to-noise ratio is
$v_\theta$. The design uses
\[
 \nu\in\{3,5,10\},\qquad
 v_\theta\in\{0.25,1\},\qquad
 \varphi\in\{0.80,0.98\}.
\]
This design varies tail thickness, state uncertainty, and persistence
while keeping observation variance fixed.
The latent state is initialized from its stationary $\N(0,v_\theta)$
distribution. All mean recursions start at zero, and the grid filter
starts from the same stationary Gaussian predictive law.

A deterministic grid filter provides the parameter-driven reference
\citep{Kitagawa:1987,DurbinKoopman:2012}. If $\pi_t^-$ and $\pi_t^+$ denote its
predictive and filtered state densities, respectively, then
\begin{align}
 \pi_t^+(\theta)&\propto p_\nu(y_t|\theta)\pi_t^-(\theta),
 \label{eq:grid-update}\\
 \pi_{t+1}^-(\theta)&=
 \int \phi_q(\theta-\varphi z)\pi_t^+(z)dz.
 \label{eq:grid-predict}
\end{align}
Write $a_t^B$, $P_t^B$, and $m_t^B$ for the predictive mean, predictive
variance, and filtered mean from this high-accuracy numerical benchmark.
The grid approximates the filtering integrals; it is not an analytically
exact filter.

The Student-$t$ likelihood makes the filtered state density non-Gaussian,
and convolution with a Gaussian transition does not generally restore
Gaussianity. We therefore use a moment-matched Gaussian law to diagnose
the one-step mean-update and Gaussian-closure errors separately. For a
Gaussian predictive law $\N(a,P)$, define the exact posterior-mean correction
\begin{equation}\label{eq:recursive-gaussian-correction}
 C^G(a,y;P)=
 \frac{\int(\theta-a)p_\nu(y|\theta)\phi_P(\theta-a)d\theta}
      {\int p_\nu(y|\theta)\phi_P(\theta-a)d\theta}.
\end{equation}
Write $D^E,D^I,D^{I\ast},D^{(2)}$ for the four corrections of
Section 5, evaluated at $(a,y;P)$.
The Gaussian-reset correction is $D^G=C^G$.
The Gaussian-reset recursion is denoted by $B$ in Section 5.8 of the main
paper and Appendix E, and by $G$ here; in this experiment, $B$ is reserved
for the grid-filter reference.
The Gaussian-observation Kalman benchmark is
$D^K(a,y;P)=P(y-a)/(1+P)$.
Each method propagates its own mean while using the common oracle
variance sequence $P_t^B$:
\begin{equation*}
 a_t^j=\varphi m_{t-1}^j,\qquad
 m_t^j=a_t^j+D^j(a_t^j,y_t;P_t^B),
 \qquad j\in\{G,E,2,I,I\ast,K\}.
\end{equation*}
Because every recursion uses the predictive variance of the grid filter,
the comparison concerns only the mean corrections, not covariance updating or
gain estimation. The experiment is therefore an oracle diagnostic and does
not evaluate a feasible forecasting procedure.

Along the grid-filter path, let $C_t^B=m_t^B-a_t^B$ and
$C_t^G=C^G(a_t^B,y_t;P_t^B)$. For any approximation $D_t^j$ evaluated at the
same moments,
\begin{equation*}
 D_t^j-C_t^B
 =\underbrace{(D_t^j-C_t^G)}_{\text{mean-update error}}
 +\underbrace{(C_t^G-C_t^B)}_{\text{predictive-shape error}}.
\end{equation*}
The first term measures approximation under the moment-matched Gaussian
law, while the second measures the error from replacing the predictive
shape. We report these oracle one-step RMSEs separately from the
accumulated RMSEs of $m_t^j-m_t^B$, since recursive feedback also changes
the incoming means at which the corrections are evaluated.

We simulate 1,000 independent stationary paths of length 500 with seed
20260815. The grid contains 481 nodes and extends eight stationary state
standard deviations in each direction. We evaluate
(\ref{eq:recursive-gaussian-correction}) with 32-node Gauss--Hermite
quadrature.

\subsection{Recursive-filter implementation and robustness}
\label{app:numerics-recursive}

The grid approximation in (\ref{eq:grid-update}) and (\ref{eq:grid-predict}) uses
an equally spaced support extending eight stationary state standard deviations
in each direction. We normalize each likelihood update separately for each
path. To propagate the density, the implementation rescales it by the AR
coefficient using linear interpolation and then convolves it with the
Gaussian innovation density on the grid.
The reported grid contains 481 nodes. The grid and quadrature checks use
auxiliary runs of 200 paths of 200 periods in two designs,
$(\nu,v_\theta,\varphi)=(3,1,0.98)$ and $(10,0.25,0.80)$, with the same
simulated observations at every resolution. In the persistent heavy-tail
design, increasing the grid to 961 nodes changes the reference filtered mean
by RMSE $4.82\times10^{-4}$ and the checked path metrics by at most
$2.46\times10^{-4}$ in absolute terms and $0.43\%$ in relative terms. In
the smooth design, the corresponding mean difference is
$3.95\times10^{-5}$. The maximum boundary mass in the full experiment is
$1.47\times10^{-10}$.

The Gaussian correction uses numerically stable Gauss--Hermite likelihood
weights. Increasing the rule from 32 to 64 nodes changes any checked
metric in these runs by at most $1.71\times10^{-6}$. We solve the implicit equation by
safeguarded Newton iterations with bisection fallback. For innovation
standard deviation $\tau$, (\ref{eq:t-derivatives}) gives the curvature
bound and sufficient condition
\begin{equation*}
 \sup_\theta H(\theta)=\frac{\nu+1}{8(\nu-2)\tau^2},
 \qquad
 \frac{P}{\tau^2}<\frac{8(\nu-2)}{\nu+1}
\end{equation*}
for global strict concavity of the Gaussian-prior log posterior, and
hence for a unique ISD mode. Every realized predictive variance satisfies
this condition with $\tau=1$. The largest values for $\nu=3,5,10$ are $1.692443$,
$1.228274$, and
$1.100503$, below the respective bounds $2$, $4$, and $64/11$. The maximum
implicit-equation residual in the full run is
$2.25\times10^{-11}$, and the minimum realized \ISDp{} curvature denominator
$1-PH(y;a+D^I)$ is $0.500011$. The largest absolute mean--mode adjustment is
$1.228043$ at the oracle moments and $1.478366$ along the actual \ISDp{}
recursion. Every realized \ISDp{} correction preserves the innovation
direction and remains between zero and the innovation.

Finite-difference checks verify all three Student-$t$ derivatives, and
the algebraic error decomposition holds numerically to $3.55\times10^{-15}$.
In an end-to-end check with $\nu=10^8$, the implemented exact
Gaussian-prior, ISD, and \ISDp{} corrections agree with the Kalman limit
$P(y-a)/(1+P)$ to within $1.18\times10^{-8}$. The explicit first-order
rule instead remains $P(y-a)$ in that limit and does not coincide
with the Kalman correction.

Table~\ref{tab:recursive-filter} reports every factorial design. For $\nu=3$,
Table~\ref{tab:recursive-t3-robust} reports absolute-error measures in addition
to RMSE because the fourth moment does not exist.
Pooled summaries give equal weight to the twelve designs.

\begin{table}[!htbp]
\tableformat
\caption{Recursive RMSE relative to the high-accuracy deterministic grid filter}
\label{tab:recursive-filter}
\begin{papertabular}{cccrrrrrr}
\toprule
$\varphi$ & SNR & $\nu$ & Gaussian reset & ESD & ISD & \ISDp{} & Explicit 2nd & Kalman \\
\midrule
0.80 & 0.25 & 3 & 0.003 & 0.100 & 0.035 & 0.007 & 0.076 & 0.162 \\
 & 0.25 & 5 & 0.001 & 0.058 & 0.017 & 0.002 & 0.022 & 0.090 \\
 & 0.25 & 10 & $6.74\times10^{-4}$ & 0.045 & 0.008 & $7.63\times10^{-4}$ & 0.010 & 0.043 \\
 & 1.00 & 3 & 0.012 & 0.480 & 0.118 & 0.037 & 0.910 & 0.266 \\
 & 1.00 & 5 & 0.007 & 0.326 & 0.069 & 0.012 & 0.345 & 0.140 \\
 & 1.00 & 10 & 0.004 & 0.273 & 0.036 & 0.005 & 0.192 & 0.061 \\
0.98 & 0.25 & 3 & 0.002 & 0.033 & 0.014 & 0.002 & 0.016 & 0.142 \\
 & 0.25 & 5 & $9.05\times10^{-4}$ & 0.022 & 0.007 & $9.59\times10^{-4}$ & 0.005 & 0.081 \\
 & 0.25 & 10 & $4.41\times10^{-4}$ & 0.017 & 0.004 & $4.44\times10^{-4}$ & 0.002 & 0.040 \\
 & 1.00 & 3 & 0.007 & 0.121 & 0.040 & 0.010 & 0.329 & 0.209 \\
 & 1.00 & 5 & 0.004 & 0.082 & 0.023 & 0.004 & 0.074 & 0.117 \\
 & 1.00 & 10 & 0.002 & 0.067 & 0.012 & 0.002 & 0.040 & 0.056 \\
\bottomrule
\end{papertabular}
\begin{floatnotes}
Notes: Entries are RMSEs of filtered posterior means relative to the
high-accuracy deterministic grid-filter mean. Each design uses all simulated periods and
independent paths. The Student-$t$ observation innovation is standardized to
unit variance, so SNR is the stationary state variance. All approximate
methods use the grid filter's predictive variance as a common oracle sequence.
Gaussian reset is the exact Gaussian-prior correction propagated recursively.
\end{floatnotes}
\end{table}

\begin{table}[!htbp]
\tableformat
\caption{Recursive errors under Student-$t_3$ observations}
\label{tab:recursive-t3-robust}
\begin{papertabular}{cclrrr}
\toprule
$\varphi$ & SNR & Method & RMSE & MAE & Median absolute error \\
\midrule
0.80 & 0.25 & Gaussian reset & 0.003 & 0.002 & 0.002 \\
 &  & ESD & 0.100 & 0.082 & 0.074 \\
 &  & ISD & 0.035 & 0.028 & 0.024 \\
 &  & \ISDp{} & 0.007 & 0.006 & 0.005 \\
 &  & Explicit 2nd & 0.076 & 0.062 & 0.056 \\
 &  & Kalman & 0.162 & 0.115 & 0.095 \\
0.80 & 1.00 & Gaussian reset & 0.012 & 0.006 & 0.004 \\
 &  & ESD & 0.480 & 0.411 & 0.389 \\
 &  & ISD & 0.118 & 0.091 & 0.075 \\
 &  & \ISDp{} & 0.037 & 0.023 & 0.012 \\
 &  & Explicit 2nd & 0.910 & 0.703 & 0.556 \\
 &  & Kalman & 0.266 & 0.165 & 0.138 \\
0.98 & 0.25 & Gaussian reset & 0.002 & 0.001 & $9.11\times10^{-4}$ \\
 &  & ESD & 0.033 & 0.025 & 0.021 \\
 &  & ISD & 0.014 & 0.010 & 0.008 \\
 &  & \ISDp{} & 0.002 & 0.002 & 0.001 \\
 &  & Explicit 2nd & 0.016 & 0.010 & 0.007 \\
 &  & Kalman & 0.142 & 0.107 & 0.087 \\
0.98 & 1.00 & Gaussian reset & 0.007 & 0.004 & 0.003 \\
 &  & ESD & 0.121 & 0.085 & 0.069 \\
 &  & ISD & 0.040 & 0.030 & 0.023 \\
 &  & \ISDp{} & 0.010 & 0.007 & 0.006 \\
 &  & Explicit 2nd & 0.329 & 0.097 & 0.055 \\
 &  & Kalman & 0.209 & 0.150 & 0.121 \\
\bottomrule
\end{papertabular}
\begin{floatnotes}
Notes: Errors are relative to the high-accuracy deterministic grid-filter mean. MAE and
median absolute error complement RMSE because standardized $t_3$ observations
have finite variance but no finite fourth moment.
\end{floatnotes}
\end{table}
\FloatBarrier

\end{document}